\documentclass[a4paper,USenglish,cleveref, autoref, thm-restate]{lipics-v2021}

\pdfoutput=1 
\hideLIPIcs  

\usepackage{enumitem}
\usepackage{algpseudocode}
\usepackage{algorithm}
\usepackage{tikz}
\usepackage{tikz-cd}
\usepackage{apxproof}
\usepackage{sidecap}
\title{Covering Simple Orthogonal Polygons with Rectangles} 

\author{Aniket Basu Roy}{Aarhus University, Denmark}{aniket@cs.au.dk}{}{}

\authorrunning{A Basu Roy}

\Copyright{Aniket Basu Roy}


\category{} 

\relatedversion{} 

\nolinenumbers 

\EventEditors{John Q. Open and Joan R. Access}
\EventNoEds{2}
\EventLongTitle{42nd Conference on Very Important Topics (CVIT 2016)}
\EventShortTitle{CVIT 2016}
\EventAcronym{CVIT}
\EventYear{2016}
\EventDate{December 24--27, 2016}
\EventLocation{Little Whinging, United Kingdom}
\EventLogo{}
\SeriesVolume{42}
\ArticleNo{23}

\usepackage{caption}
\usepackage{subcaption}
\usepackage[T1]{fontenc}
\usepackage{graphicx}
\usepackage{url}
\usepackage{xspace,framed}

\newcommand{\why}{{\color{red} \LARGE\textbf{Why?}\xspace}}

\newcommand{\defproblem}[3]{
  \vspace{3mm}
\noindent\fbox{
  \begin{minipage}{.95\textwidth}
  \begin{tabular*}{\textwidth}{@{\extracolsep{\fill}}lr} #1  \\ \end{tabular*}
  {\bf{Input:}} #2  \\
  {\bf{Question:}} #3
  \end{minipage}
  }
  \vspace{2mm}
  }

\newcommand{\FF}{{\mathcal F}}

\newcommand{\PP}{{\mathcal P}}

\newcommand{\RR}{{\mathcal R}}

\newcommand{\TT}{{\mathcal T}}

\newcommand{\Real}{\ensuremath{\mathbb{R}}\xspace}%
\newcommand{\Natural}{\ensuremath{\mathbb{N}}\xspace}%

\newcommand{\dP}{\partial P}

\usepackage[usenames,dvipsnames]{pstricks}
\usepackage{caption}
\usepackage{subcaption}
\usepackage{thmtools, thm-restate} 

\newcommand{\hide}[1]{}

\newenvironment{hproof}{%
  \renewcommand{\proofname}{Proof Sketch}\proof}{\endproof}
\begin{document}

\maketitle


\begin{abstract}
  We study the problem of Covering Orthogonal Polygons with Rectangles. For polynomial-time algorithms, the best-known approximation factor is $O(\sqrt{\log n})$ when the input polygon may have holes [Kumar and Ramesh, STOC '99, SICOMP '03], and there is a $2$-factor approximation algorithm known when the polygon is hole-free [Franzblau, SIDMA '89]. Arguably, an easier problem is the \emph{Boundary Cover} problem where we are interested in covering only the boundary of the polygon in contrast to the original problem where we are interested in covering the interior of the polygon, hence it is also referred as the \emph{Interior Cover} problem. For the \emph{Boundary Cover} problem, a $4$-factor approximation algorithm is known to exist and it is \emph{APX}-hard when the polygon has holes [Berman and DasGupta, Algorithmica '94].
  
  In this work, we investigate how effective is local search algorithm for the above covering problems on simple polygons.
  We prove that a simple local search algorithm yields a \emph{PTAS} for the \emph{Boundary Cover} problem when the polygon is simple. Our proof relies on the existence of planar supports on appropriate hypergraphs defined on the \emph{Boundary Cover} problem instance.
  On the other hand, we construct instances where support graphs for the \emph{Interior Cover} problem have arbitrarily large bicliques, thus implying that the same local search technique cannot yield a PTAS for this problem. We also show large locality gap for its dual problem, namely the \emph{Maximum Antirectangle} problem.
\end{abstract}

\newpage
\section{Introduction}

Polygon Covering is a classical problem in Computational Geometry.
In this work we are interested in orthogonal polygons, i.e., all its sides are horizontal or vertical line segments, and we are required to cover the polygon with a minimum number of orthogonal rectangles such that the union of the rectangles is same as the polygon. This problem has a long history and
it has been studied under the guise of various names.
To the best of our knowledge, the earliest reference to this problem is an unpublished work by Masek \cite{masek1979some} which is mentioned by Garey and Johnson in their book \emph{Computers and Intractability}
\cite{GareyJohnson79}.
Here, the problem is referred to as the Rectilinear Picture Compression problem 
and claimed to be $NP$-hard.
Later Culberson and Reckhow \cite{culberson1994covering} proved it to be $NP$-hard for simple polygons, i.e., polygons that do not have holes.

Orthogonal Polygon Covering problem has been studied also from the perspective of polynomial time approximation algorithms where the state-of-the-art approximation ratio is $O(\sqrt{\log n})$ by Kumar and Ramesh \cite{kumar2003covering}. 
Also, the question has been posed for simple polygons. Franzblau gave a $(2+h)$-approximation algorithm which implies a $2$-factor approximation for simple polygons \cite{franzblau1989performance}, where $h$ is the number of holes.
Whereas in 1981, Chaiken, Kleitman, Saks, and Shearer~\cite{chaiken1981covering} gave a polynomial time exact algorithm when the polygon is both vertically and horizontally convex. This restriction was relaxed to just vertically (resp. horizontally) convex polygons to attain the same result by Franzblau and Kleitman \cite{FRANZBLAU1984164}.

Chaiklen et al. \cite{chaiken1981covering} also studied the problem of covering the boundary of the polygon (which they referred as \emph{Edge Cover} problem) and the corners of the polygon with minimum number of rectangles, and studied the respective duality gaps.
Culberson and Reckhow \cite{culberson1994covering} proved that even the boundary cover problem is $NP$-hard for hole-free polygons. Berman and DasGupta \cite{berman1997complexities} showed that the problem is $APX$-hard when we allow holes in the polygon for the interior cover as well as the boundary cover problem.

Although several articles have been published over the years, especially from the perspective of applications \cite{combarro2023constraint,koch2023hybrid}, theoretical improvements have stagnated in the last two decades.

In the current work, we are interested in giving better approximation algorithms for the boundary cover problem for simple orthogonal polygons.
To our knowledge, the best known approximation factor is $4$ by the same algorithm that works for the boundary cover problem for arbitrary polygons \cite{berman1997complexities}.

\begin{center}
  \begin{table}[h]
    \caption{State of the Art for a few variants of the Polygon Covering problem. Our results are highlighted in bold.}
  \begin{tabular}{|l|l|l|}
      \hline
      & Without holes & With holes \\
      \hline
      & & \\
      Corner & $2$-apx \cite{berman1997complexities}, {\bf PTAS} & $NP$-h \cite{berman1997complexities}, $4$-apx \cite{berman1997complexities} \\
      & & \\
      Boundary & $NP$-h \cite{culberson1994covering}, {\bf PTAS} & $NP$-h \cite{conn1987some}, $4$-apx \cite{berman1997complexities}, $APX$-h \cite{berman1997complexities} \\
      & & \\
      Interior & $NP$-h \cite{culberson1994covering}, $2$-apx \cite{franzblau1989performance} & $NP$-h \cite{masek1979some}, $O(\sqrt{\log n})$-apx \cite{kumar2003covering}, $APX$-h \cite{berman1997complexities} \\
      \hline
  \end{tabular}
  \label{tab:soa}
  \end{table}
\end{center}

\subsection{Our Contribution}
Our main contribution is to show the existence of planar support for appropriate hypergraphs defined on the boundary cover problem instance as stated below.


\begin{restatable}[Planar Support for Boundary Cover]{theorem}{BoundaryPlanarSupport}
  \label{thm:main}
  Given a simple orthogonal polygon $P$ and a set of containment-maximal rectangles $\RR$ with respect to $P$, there exists planar support graph for $(\RR,\dP)$.
\end{restatable}

In order to prove this theorem, we observe many novel properties of containment-maximal rectangles, which could be of independent interest and useful elsewhere.
As a corollary of this theorem, it immediately follows that local search yields a Polynomial Time Approximation Scheme (PTAS) for the boundary cover problem and the corner cover problem by applying Lemma~\ref{lem:support-implies-ptas}, which is a framework, independently devised by Chan and Har-Peled \cite{Chan2012} and Mustafa and Ray \cite{MR10}. Please refer to Subsection~\ref{sec:local-search} for more on this framework.

\begin{corollary}
  \label{cor:boundary}
  Local Search yields a PTAS for the Boundary Cover problem for simple orthogonal polygons.
\end{corollary}

\begin{corollary}
  \label{cor:corner}
  Local Search yields a PTAS for the Corner Cover problem for simple orthogonal polygons.
\end{corollary}

On the negative side, we construct interior cover problem instances whose minimal support graphs contain arbitrarily large bicliques thus implying that a PTAS cannot be attained using the above local search framework (Theorem~\ref{thm:interior}).
We also give a negative result for the dual problem of finding maximum sized antirectangle (also known as the maximum hidden set problem) by showing the existence of instances such that the dual intersection graphs have arbitrarily large bicliques, thus implying local search can be arbitrarily smaller than the optimum solution (Corollary~\ref{cor:antirectangle}).
Please refer to Section~\ref{sec:prelims} for the precise definitions of the above problems.

\subsubsection{High-level Idea and Organization}
There are broadly two parts to our proof of the existence of the planar supports for hypergraphs defined on the boundary cover instances. If we think of the existence of planar supports as the desired property then we prove that this property is \emph{hereditary}, i.e., if a family satisfies the property then so does any of its sub-families (Lemma~\ref{lem:hereditary-planar} in Section~\ref{sec:hereditary}).  
We also prove that the \emph{complete} families of containment-maximal rectangles have this property (Lemma~\ref{thm:complete-rect} in Section~\ref{sec:complete}). By complete we mean the set of all possible containment-maximal rectangles with respect to $P$. 

In Section~\ref{sec:hereditary}, we use the notion of Left-Right Coloring to prove the hereditary property. The Left-Right Coloring has been used for testing planarity of graphs in the literature~\cite{de2006tremaux}. The central idea in proving the hereditary property is that if a family has a star graph as its support then the family without the rectangle corresponding to the center of the star, still has a planar support. See Figure~\ref{fig:biclique-v1} for an example and Section~\ref{sec:star} for the technical details.

We prove the planar support property for complete family of rectangles by induction on the number of vertices of the horizontal R-tree. Given an orthogonal polygon $P$, we prove that if the statement is true for a complete family with respect to a \emph{smaller} polygon, then it is also true for the complete family with respect to $P$. The formal statement can be found in Section~\ref{sec:complete} and the full proof can be found in Section~\ref{sec:complete-long} in the appendices.

\begin{figure}[h]
  \centering
  \begin{subfigure}{0.45\textwidth}
    \centering
    \includegraphics[width=0.9\textwidth]{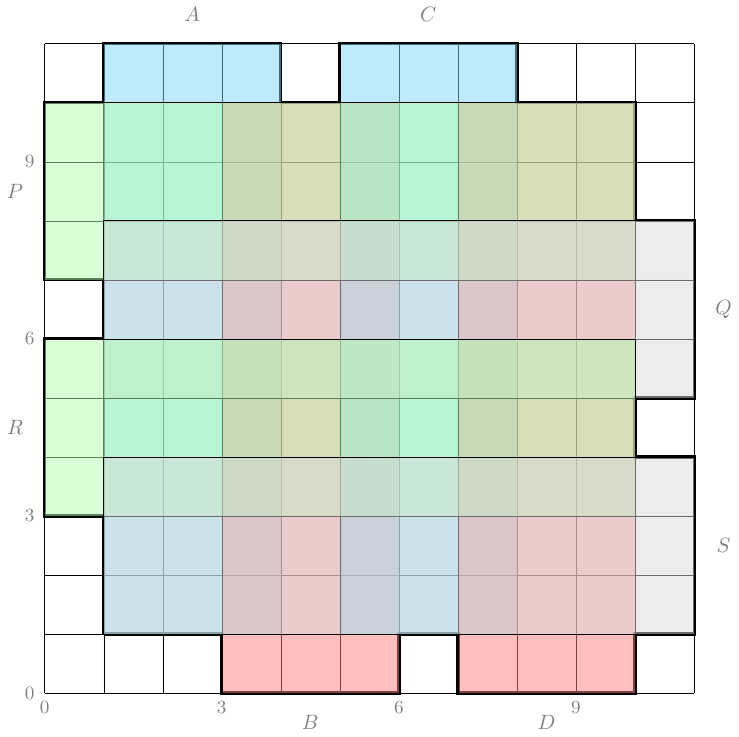}
  \end{subfigure}
  \begin{subfigure}{0.45\textwidth}
    \centering
    \includegraphics[width=1.25\textwidth]{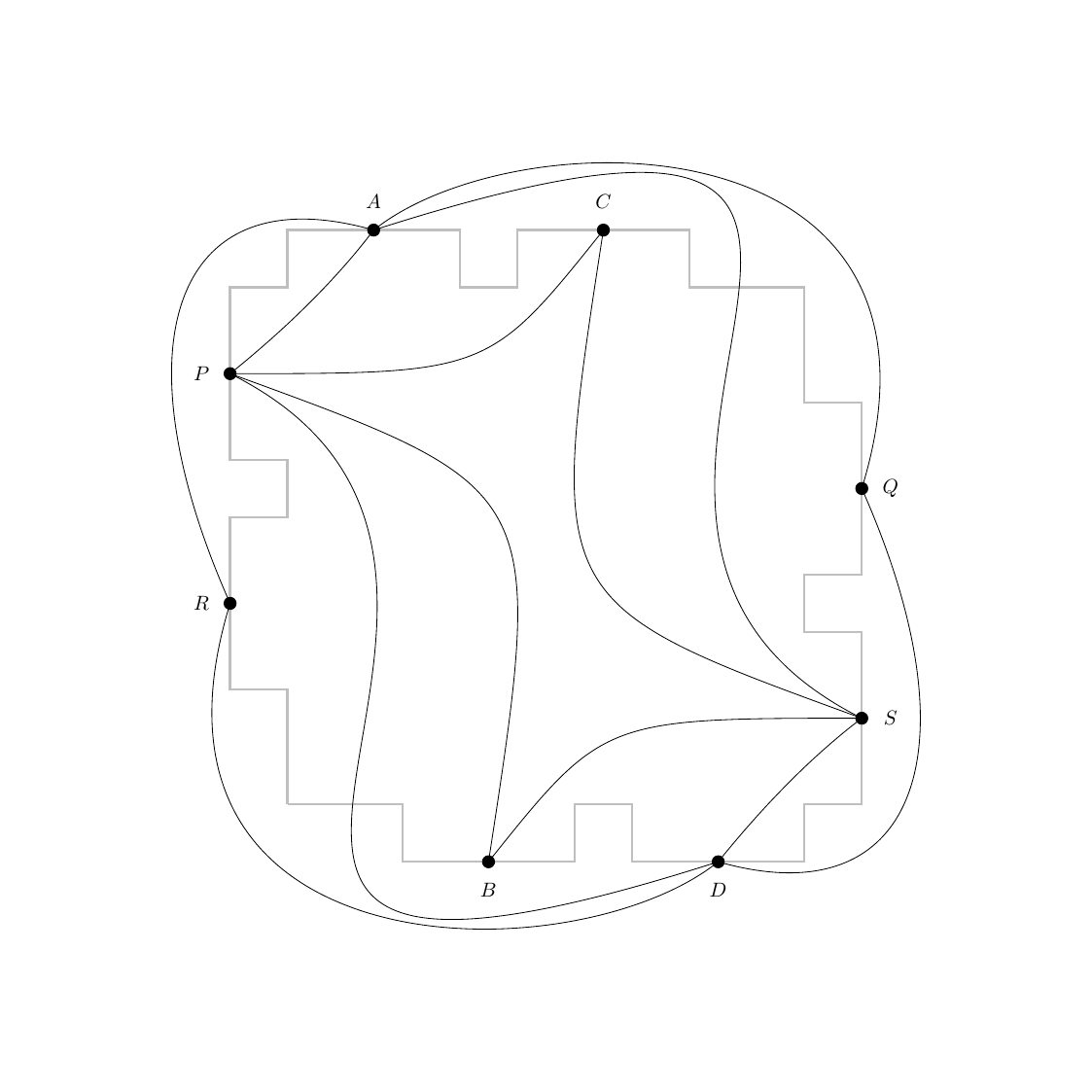}
  \end{subfigure}
  
        \caption{On the left, the boundary of polygon $P$ is drawn with thick black lines on an integer grid. $\RR = \{A,B,C,D,P,Q,R,S\}$ where
          $A = [1,4]\times[1,11]$,
          $B = [3,6]\times[0,10]$,
          $C = [5,8]\times[1,11]$,
          $D = [7,10]\times[0,10]$,       
          $P = [0,10]\times[7,10]$,
          $Q = [1,11]\times[5,8]$,
          $R = [0,10]\times[3,6]$,
          $S = [1,11]\times[1,4]$.
          A support graph for the boundary cover instance $\RR$ with respect to $P$ is drawn on the right.
          Let $R_c = [1,10]\times[1,10]$. The minimal support graph for the boundary cover instance $\RR\cup\{R_c\}$ is a star graph with $R_c$ as the center of the star.
          }
        \label{fig:biclique-v1}
\end{figure}


Apart from this, the list of our results can be found in Section~\ref{sec:main-results}.
We discuss some related work in Subsection~\ref{sec:rel-work}. Section~\ref{sec:prelims} deals with some preliminary definitions and results required to build up the later sections.
In particular, Subsection~\ref{sec:lrcolor} is devoted to Left-Right Coloring.
This section supports the ideas discussed in Section~\ref{sec:hereditary}.
A number of proofs and definitions related to Sections~\ref{sec:star} and~\ref{sec:complete} that does not appear in the main body of the article can be found in the appendices.
Finally, we discuss a few open problems in Section~\ref{sec:open}.

\subsection{Related Work}
\label{sec:rel-work}

\subsubsection{Visibility Graphs}
\label{sec:visi-graph}

Given an orthogonal polygon $P$, define a graph $G = (V,E)$ where $V$ is finite point set in $P$, and $(p,q) \in E(G)$ if there exists an axis-parallel rectangle $R \subseteq P$ such that $p,q \in R$.

Observe that a polygon covering corresponds to a clique cover in $G$ if $V$ is sufficiently dense in $P$. A natural lower bound is the size of an independent set in $G$, which is also called an antirectangle. The minimum clique cover size is denoted by $\theta$ and the maximum antirectangle (independent set) size is denoted by $\alpha$.

As mentioned by Chaiken et al.~\cite{chaiken1981covering},
Chv\'{a}tal asked whether $\theta \stackrel{?}{=} \alpha$, which was disproved by Szemer\'{e}di
by constructing a polygon with $\theta =8$ and $\alpha = 7$. Chung managed to attain the same bound using a simple polygon. Erd\H{o}s asked whether there exists a constant $c > 1$ such that $\theta \le c \cdot \alpha$. The only known upper bound is $\theta = O(\alpha \log \alpha)$ due to Franzblau \cite{franzblau1989performance}.

\subsubsection{Rectangle Rule List Minimization}
Applegate et al.~\cite{applegate2007compressing} studied the
Rectangle Rule List Minimization problem, which is to find the
shortest list of rules needed to create a given pattern. Given a  rectangular grid colored $c_0$, a pair $(R,c)$ is called a \emph{rectangle rule} where $R$ is a rectangle and $c$ is a color from a color palette. A sequence of rectangle rules $\RR$, $(R_1, c_1), \dots, (R_n, c_n)$ is called a Rectangle Rule List, and $P_\RR$ is the colored pattern such a sequence produces. The objective is to find a minimum length sequence $\RR$ that creates a target colored pattern. When the color palette has a single color, it is the Orthogonal Polygon Covering problem.

\subsubsection{Covering Boolean Matrices with Combinatorial Rectangles}

As orthogonal polygons can be drawn in a bounded grid, one can map every grid cell to $1$ if it is contained within the input polygon and $0$ otherwise.
In Communication Complexity, the following problem is of considerable interest.

Given an $n \times n$ Boolean matrix $M = A \times B$, where $A, B \in \{0,1\}^n$, a combinatorial rectangle $R$ is defined as $R = K \times L \subseteq [n] \times [n]$. A $1$-rectangle is a combinatorial rectangle $R$ such that all its entries are $1$. Finding the minimum number of $1$-rectangles such that all the $1$s in $M$ are covered is called the rectangle covering problem.
Note that in this context one is free to permute the rows and columns unlike our case.
 It has the following implication with the non-deterministic complexity of $M$.


\begin{proposition}[Proposition 5.2 in \cite{roughgarden2016communication}]
  Let $f: X \times Y \rightarrow \{0,1\}$ and $M(f)$ be the corresponding Boolean matrix. If there is a cover of $1$-entries of $M(f)$ by $t$ $1$-rectangles, then there is a non-deterministic protocol that verifies whether $f(x,y) = 1$ with cost $\log_2 t$. 
\end{proposition}

This problem has also been referred as the Boolean Basis problem and related to the Polygon Covering problem by Lubiw \cite{lubiw1990boolean}. Yet another way to view the problem is to cover the edges of the bipartite graph, whose adjacency matrix is $M$, with minimum number of bicliques.

\section{Preliminaries}
\label{sec:prelims}

All polygons are assumed to be orthogonal, i.e., their sides are parallel to either the vertical or horizontal axes. A polygon is said to be simple if it has no holes, i.e., it is simply connected. Given a simple orthogonal polygon $P$, we denote $\partial P$ to be the boundary of $P$. All rectangles are assumed to be axis-parallel.
A rectangle $R \subseteq P$ is called containment-maximal with respect to a polygon $P$ if $R$ is not contained in a rectangle $R' \subseteq P$ of larger area. Henceforth, we shall use the word ``maximal'' to denote containment-maximal for the sake of brevity. 
Observe, as a consequence of maximality, all four sides of $R$ have non-empty intersection with $\dP$.
We denote $\RR^K(P)$ to be the set of all maximal rectangles with respect to $P$, also referred as the complete family of rectangles with respect to $P$.

\begin{definition}[Blocker of a Maximal Rectangle]
Given an orthogonal polygon $P$ and a maximal rectangle $R \subseteq P$, any point $p \in \dP$ is called a blocker of $R$ if $p \in R$ and $p$ is not a corner of $R$. If $p$ intersects the top (resp. bottom, left, right) side of $R$ then $p$ is called a top (resp. bottom, left, right) blocker of $R$.
\end{definition}

For the sake of brevity, we shall sometimes call a side $s$ of $P$ to be a blocker of $R$ if $s \cap R \neq \emptyset$.

Given a family of maximal rectangles $\RR$ contained in $P$, 
for every point $p \in \partial P$ we define a hyperedge $\RR_p := \{ R \in \RR \mid p \in R\}$.
We shall work with this hypergraph $(\RR, \{\RR_p\}_{p \in \dP})$ where the vertex set is $\RR$ and the hyperedges are $\RR_p$ for every $p \in \partial P$. 
For the sake of brevity, we shall denote $(\RR, \{\RR_p\}_{p \in \dP})$ as $(\RR, \partial P)$. Similarly, one can define $(\RR, P)$.

\begin{definition}[Support graph of a hypergraph]
	Given a hypergraph $H = (V,\FF)$ where $V$ is the vertex set and $\FF \subseteq 2^V$ is a family of subsets of $V$, we say $G = (V,E)$ is a support graph of $H$ if for every $F \in \FF$, the subgraph induced on F, $G[F]$ is connected.
\end{definition}

See Figure~\ref{fig:biclique-v1} for an example of a family of maximal rectangles $\RR$ with respect to polygon $P$ on the left, and a planar support graph for the hypergraph $(\RR, \dP)$ to its right. Observe, that the intersection graph of the rectangles is always a valid support graph, which may not be planar.
      
      \defproblem
      {\textsc{Interior Cover problem}}
      {An orthogonal polygon $P$ in the plane.}
      {A minimum sized set of rectangles $\RR$ such that $$\bigcup_{R_i \in \RR} R_i = P.$$ }
             
      \defproblem
      {\textsc{Boundary Cover problem}}
      {An orthogonal polygon $P$ in the plane.}
      {A minimum sized set of rectangles $\RR$ such that  $$\partial P \subseteq \bigcup_{\substack{R_i \in \RR \\ R_i \subseteq P}} R_i.$$}

      \defproblem
      {\textsc{Corner Cover problem}}
      {An orthogonal polygon $P$ in the plane.}
      {A minimum sized set of rectangles $\RR$ such that  $$\textsf{corner} P \subseteq \bigcup_{\substack{R_i \in \RR \\ R_i \subseteq P}} R_i$$
        where $\textsf{corner} P$ denotes the set of corners of $P$.}
      
      \defproblem
      {\textsc{Maximum Antirectangle problem}}
      {An orthogonal polygon $P$ in the plane.}
      {A maximum sized set of points $Q \subset P$ such that for every $p,q \in Q$, there does not exist a rectangle $R \subseteq P$ where $p,q \in R$.}




\subsection{Local Search Framework}
\label{sec:local-search}

Local search has been a very useful heuristic for decades, and has been proved to give good approximation ratios, particularly for clustering and geometric optimization algorithm \cite{AGKMP01,abs-0809-2554}. In this work, we will consider the particular local search framework introduced by Mustafa and Ray for covering problems in \cite{MR10}. Almost the same framework works for packing problems, which was introduced around the same time by Chan and Har-Peled in \cite{Chan2012}. Over the years, this framework has been used for a variety of packing and covering problems \cite{aschner2013approximation,ashok2020local,ashur2022terrain,basu2018packing,Har-PeledQ15,krohn2014guarding,raman2020constructing,raman2023hypergraph}. As a consequence of this framework, it is sufficient to show the existence of a sparse support graph for a hypergraph: in order to prove that local search yields a PTAS for the corresponding covering problem. Here by sparse, we mean the graph has strongly sublinear separators, which is a necessary condition for planarity.

The study of planar supports predates the local search framework in question. It has been studied in the context of hypergraph drawing. See \cite{buchin2011planar,johnson1987hypergraph,kaufmann2008subdivision}, and Section 2 in \cite{raman2023hypergraph} for a short survey.
Durocher and Fraser \cite{durocher2015duality} were arguably the first to make the connection and use the word ``support'' in the context of local search for covering problems.

Below, we state an abstraction of the result of Mustafa and Ray \cite{MR10}.
\begin{lemma}[Local Search Lemma \cite{MR10}]
  \label{lem:support-implies-ptas}
  Given a hypergraph $H = (V,\FF)$ where $V$ is the vertex set and $\FF \subseteq 2^V$ is a family of subsets of $V$, we want to compute the minimum sized $V' \subseteq V$ such that for every $F \in \FF$, $V' \cap F \neq \emptyset$. If there exists a planar support of $H$, then the local search algorithm computes a $(1+\epsilon)$-approximate solution in time $n^{O(1/\epsilon^2)}$, for $\epsilon >0$.
\end{lemma}

Lastly, we state the local search algorithm for covering the boundary of the input simple polygon $P$, where $k$ is the local search parameter which is set to $c/\epsilon^2$ for some appropriate constant $c$.

 \begin{algorithm}
    \caption{Local Search Algorithm to cover the boundary of $P$.}
    \label{alg:ls}  
    \begin{algorithmic}[1]
      \Require A simple polygon $P \subset \Real^2$ and $k \in \Natural$.
      \Ensure $\RR = \{R_1, \dots, R_m\}$ such that $\dP \subseteq \bigcup \RR$.
      \State $S \gets \RR^K(P)$\;
      \State $S^c \gets \RR^K(P) \setminus S$\;
      \While {$\exists X \subseteq S, \exists Y \subseteq S^c$ such that $|X| \le k$, $|Y| < |X|$, and $(S\setminus X) \cup Y$ covers $\dP$}
      \State $S \gets (S\setminus X) \cup Y$\;
      \State $S^c \gets \RR^K(P) \setminus S$\;
      \EndWhile
      \State \Return $S$\;
    \end{algorithmic}
    \end{algorithm}


\subsection{Intersection Patterns}

Given a simple orthogonal polygon $P$, if two maximal rectangles $R$ and $R'$ with respect to $P$ have a non-empty intersection then they intersect in one of the following ways:
\begin{enumerate}
\item {\bf Corner Intersection.} $R$ has exactly one corner of $R'$ in its interior and vice versa.
\item {\bf Piercing Intersection.} The two vertical (resp. horizontal) sides of $R$ intersect the two horizontal (resp. vertical) sides of $R'$.
\end{enumerate}

Furthermore, if $R$ and $R'$ have a piercing intersection and one of the sides of $R$ is contained in one of the sides of $R'$ then we say that the two rectangles are \emph{aligned} to each other. See Figure~\ref{fig:rect-intersect}.

 \begin{figure}[h]
        \centering
        \includegraphics[width=0.5\textwidth]{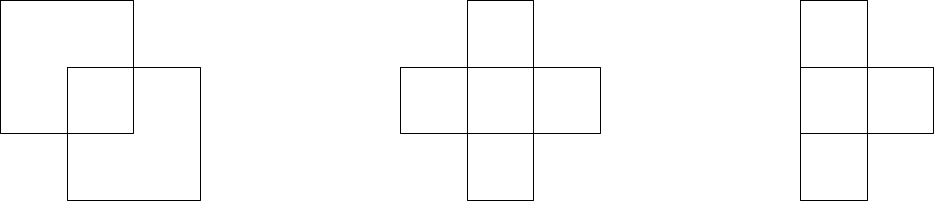}
        \caption{The left pair has a corner intersection. The middle and the right have a piercing intersection; in which the right pair are also aligned to each other.}
        \label{fig:rect-intersect}
    \end{figure}

\begin{definition}
  \label{defn:pierce-poset}
  Given a simple orthogonal polygon $P$ and a set of maximal rectangles $\RR$ with respect to $P$, define a partial order $(\RR,\prec)$. If for every $R ,R' \in \RR$ such that $R$ and $R'$ have a piercing intersection and the two vertical sides of $R$ intersect the two horizontal sides of $R'$ then we say $R \prec R'$.
\end{definition}

More informally, we say $R \prec R'$ if they have a piercing intersection and the width of $R$ is less than that of $R'$.

Lastly, we would like to emphasize that we shall deal with intersections at the boundary of the polygon, i.e., $R\cap R' \cap \dP \neq \emptyset$ which may not be true even when $R \cap R' \neq \emptyset$.

\subsection{Left-Right Coloring}
\label{sec:lrcolor}

In this section we recall a few definitions and results from a series of work by de Fraysseix, Ossona de Mendez, and Rosenstiehl \cite{de2006tremaux}.

\begin{definition}[DFS-orientation]
Given an undirected graph $G = (V,E)$, its DFS-orientation is a directed graph $G' = (V, T \uplus B)$ where $T$ is the set of the tree edges obtained by depth-first traversal of $G$ starting from a fixed vertex in $V$, and $B = E\setminus T$ is the set of cotree edges.
\end{definition}

Thus, $T$ defines a partial order on $V$ and the root of $T$ is the minimum vertex.
Given a tree edge $e = (u, v)$, return edges of $e$ are the cotree edges $e' = (v', u')$ such that $v'$ is a descendant of $v$ or $v = v'$, and $u'$ is an ancestor of $u$ that is different from $u$.
Define a mapping $low: E \rightarrow V$, where for every $e \in E$, $low(e)$ is the minimum vertex that can be reached by a directed path starting with $e$ and containing at most one cotree edge.

\begin{definition}[Left-Right Coloring]
Let $G' = (V, T \uplus B)$ be a DFS-oriented graph. A bipartition of its co-tree edges $B = L \uplus R$ into two classes, referred to as left and right, is called left-right coloring,
  if for every $v \in V$ and outgoing edges $e_i$ and $e_j$ from $v$, the following two conditions are satisfied.
  \begin{itemize}
  \item
    All return edges of $e_i$ ending strictly higher than $low(e_j)$ belong to one class, and
  \item
    All return edges of $e_j$ ending strictly higher than $low(e_i)$ belong to the other class.
  \end{itemize}
\end{definition}

\begin{lemma}{\cite{de2006tremaux}}
  \label{lem:left-right}
  A finite graph $G$ is planar if and only if there exists a Left-Right coloring of the cotree edges with respect to a depth-first-tree $T$ on $G$.
\end{lemma}

\section{Family of Rectangles with a Kernel}
\label{sec:star}

The main results of this section are Lemmata~\ref{lem:star-kernel} and \ref{lem:etoile-sans-centre}, which in turn are used to prove Lemma~\ref{lem:hereditary-planar} in Section~\ref{sec:hereditary}.
(Please see Section~\ref{sec:star-appendix} for missing details.)
In the following, we define a kernel of a rectangle family.
Informally, Lemma~\ref{lem:star-kernel} says that if a rectangle family has a star graph as its support, then the rectangle corresponding to the center of the star belongs to the kernel. Lemma~\ref{lem:etoile-sans-centre} states that a family of rectangles without its kernel still has a planar support.

\begin{definition}[Kernel of a Family]
  Given a family of maximal rectangles $\RR$ with respect to a simple polygon $P$, the kernel of $(\RR, \dP)$ is defined as follows.
  $$ker(\RR) := \displaystyle\bigcap_{\substack{p \in \dP \\ |\RR_p| \ge 2}} \RR_p$$
  where $\RR_p:= \{R \in \RR \mid p \in R\}$.
\end{definition}

A family $(\RR, \dP)$ with non-empty kernel $ker(\RR)$ is called \emph{proper} if for every $R \in \RR$ there exists $p \in \dP$ such that $ker(\RR) \cup \{R\} \subseteq \RR_p$. See Figure~\ref{fig:proper-kernel} for an example of a family that is not proper.

\begin{figure}[h]
        \centering
        \includegraphics[width=0.6\textwidth]{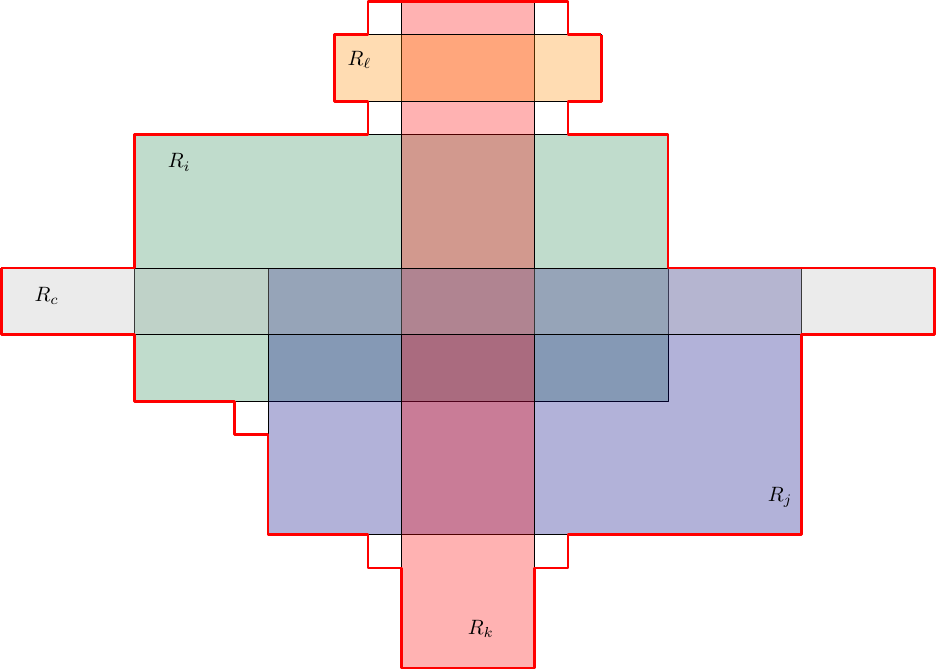}
        \caption{$(\{R_i, R_j, R_k, R_\ell, R_c\}, \dP)$ has a non-empty kernel $\{R_c\}$ but it is not proper. Whereas, $(\{R_i, R_j, R_c\}, \dP)$ is proper and has the same kernel. Note that $\dP$ is drawn with red lines.}
        \label{fig:proper-kernel}
\end{figure}




\begin{observation}
  \label{lem:kernel-star}
  Every proper family $(\RR, \dP)$ with a non-empty kernel has a support that is a star graph with the center vertex corresponding to a kernel rectangle. 
\end{observation}

In the following, we show that the converse is also true.

\begin{lemma}
  \label{lem:star-kernel}
Given a family of maximal rectangles $\RR$ with respect to a simple orthogonal polygon $P$. If there exists a minimal support graph $G$ of $(\RR, \dP)$ such that $G$ is a star graph then the rectangle corresponding to the center of $G$ belongs to $ker(\RR)$ and $\RR$ is proper.
\end{lemma}

\begin{proof}
  Let $R_c$ be the center of $G$. As $G$ is a minimal support graph of $(\RR, \dP)$, for every $p \in \dP$ such that $|\RR_p| \ge 2$, $R_c \in \RR_p$. If there exists $p \in \dP$ such that $|\RR_p| \ge 2$ and $R_c \notin \RR_p$ then $p$ is not supported by $G$ and hence a contradiction. This implies $R_c \in ker(\RR)$.

  $\RR$ is proper because otherwise there are isolated vertices in a minimal support. 
\end{proof}



Throughout this section we assume $\RR$ to be proper and have a non-empty kernel. Without loss of generality, we consider $\RR$ to have a unique kernel rectangle, i.e., $ker(\RR) = \{R_c\}$.
From Definition~\ref{defn:pierce-poset}, we can partition the family without its kernel rectangle $R_c$ into the following subsets. The set of rectangles that have a corner intersection with $R_c$ that we denote by $\RR_N$. The remaining rectangles in $\RR$ have a piercing intersection with $R_c$. Among the ones having a piercing intersection, we call the rectangles $R \in \RR$ to be vertical rectangles if $R \prec R_c$, i.e., $\RR_V := \{R \in \RR \mid R \prec R_c\}$. Likewise, we call the rectangles $R \in \RR$ to be horizontal rectangles if $R_c \prec R$, i.e., $\RR_H := \{R \in \RR \mid R_c \prec R \}$.
Thus, $\RR := \RR_N \uplus \RR_V \uplus \RR_H \uplus \{R_c\}$.

Finally, we state the remaining main lemma of this section.
  \begin{restatable}{lemma}{EtoileSansCentre}
  \label{lem:etoile-sans-centre}
Given a family $(\RR, \dP)$ such that $ker(\RR) \neq \emptyset$, then $(\RR\setminus ker(\RR), \dP)$ has a planar support.
  \end{restatable}

  \begin{hproof}
    The proof to this lemma can be found in Section \ref{sec:star-appendix}. Here we sketch the proof.
    To this end we prove a few statements regarding how the rectangles from $\RR_N, \RR_V,$ and $\RR_H$ intersect among themselves. Furthermore among the rectangles in $\RR_H$ (resp. $\RR_V$), we denote the set of minimally horizontal (resp. maximally vertical) rectangles as $\RR^*_H$ (resp. $\RR^*_V$). Then we show that the rectangles in $\RR^*_H$ (resp. $\RR^*_V$) are in a linear order $<_h$ (resp. $<_v$).
    This gives us enough understanding to design an algorithm to draw a planar graph that is also a support graph of $(\RR\setminus ker(\RR), \dP)$.

    We very briefly give an overview of our algorithm in the following. We place the vertices corresponding to the minimally horizontal rectangles, i.e., $\RR^*_H$ on a line in the same order as in $<_h$. Similarly we do the same for the rectangles in $\RR^*_V$.
Every vertex corresponding to a vertical (resp. horizontal) rectangle is made adjacent to a vertex corresponding to an aligned horizontal (resp. vertical) rectangle.
See Figure~\ref{fig:kernel-schema} for a schema of the drawing of the graph $G$.
Refer to Subsection~\ref{sec:algo-kernel-less} in the appendices for the complete algorithm to draw planar support $G$ for $\RR\setminus ker(\RR)$.
  \end{hproof}


   \begin{SCfigure}[1][h]
        \centering
        \includegraphics[width=0.45\textwidth]{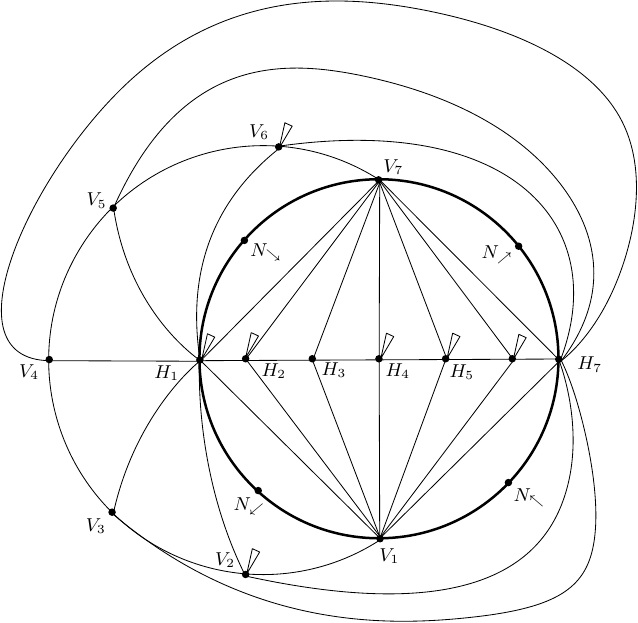}
        \caption{Schema of a planar support graph for $(\RR\setminus ker(\RR), \dP)$. $\RR^*_V = \{V_1,\dots, V_7\}$, $\RR^*_H = \{H_1, \dots, H_7\}$, $\RR_N = \{N_\swarrow, N_\nwarrow, N_\nearrow, N_\searrow\}$. For every $V_i \in \RR^*_V$, if there exists $R \prec V_i$, the vertex corresponding to $R$ is placed in the private triangular region attached to $V_i$. Likewise for vertices in $\RR^*_H$.}
        \label{fig:kernel-schema}
      \end{SCfigure}



\section{Hereditary Property of Planar Boundary Supports}
\label{sec:hereditary}

In this section, we prove Lemma~\ref{lem:hereditary-planar}, i.e., the property of the boundary supports being planar is hereditary, which in turn is used to prove Theorem~\ref{thm:main}. To that end, we need Lemmata~\ref{lem:left-right}, \ref{lem:star-kernel}, and~\ref{lem:etoile-sans-centre}.


\begin{lemma}[Hereditary Property]
  \label{lem:hereditary-planar}
  Given a simple orthogonal polygon $P$ and a set of maximal rectangles $\RR$ with respect to $P$. If $(\RR,\dP)$ has a planar support then for every $\RR' \subset \RR$, $(\RR',\dP)$ also has a planar support.
\end{lemma}

\begin{proof}
  We are given a planar support graph $G$ for $(\RR,\dP)$. We claim that for every $R \in \RR$, there exists a planar support graph $G'$ for $(\RR\setminus\{R\},\dP)$.
  If this claim is true, then for every $\RR' \subset \RR$, $(\RR',\dP)$ also has a planar support.

  Fix some $R \in \RR$ and let $v$ be the vertex in $G$ corresponding to $R$. In short, we denote $R = R(v)$.  Recall the definition of left-right coloring (Subsection~\ref{sec:lrcolor}).
Consider a DFS tree $T$ defined on $G$ rooted at $v$. Let $w_i$ be a child of $v$ in $T$, for every $i \in [t]$. We denote $T(w_i)$ to be the subtree rooted at $w_i$ in $T$. Let $(x_{i,j}, v)$ be a cotree edge such that it points from $x_{i,j}$ to $v$, where $x_{i,j} \in T(w_i)$.

From Lemma~\ref{lem:left-right}, we know that if $G$ is planar, then there exists a left-right coloring of the cotree edges. 
We do the following \emph{shortcutting} operation on the cotree edges $(x_{i,j},v)$.
We replace the edge $(x_{i,j}, v)$ by $(x_{i,j}, w_i)$ such that the color class of the cotree edge, in the aforesaid left-right coloring, is same as before after doing the shortcutting.
See Figure~\ref{fig:dfs}.
From Lemma~\ref{lem:left-right}, it follows that planarity is preserved. Let $E_1$ be the set of edges that we get after the shortcutting operation.

  For every $p \in \dP$, if $\RR_p$ lies in some subtree $T(w_i)$, for some $i \in [t]$, then $\RR_p$ is supported by this shortcutting. If the vertices corresponding to the rectangles in $\RR_p$ are distributed across the subtrees rooted at multiple children of $v$ then observe the following.
  \begin{observation}
    If there exist vertices $x,y \in \RR_p$ such that $x \in T(w_i)$ and $y \in T(w_j)$, for $i,j \in [t]$ and $i \neq j$, then $w_i, w_j \in \RR_p$.
  \end{observation}

  This is true because otherwise point $p$ is not supported by $G$.
  Therefore, we can ignore the descendants of $w_i$, for every $i \in [t]$ and only put edges among the $w_i$'s in order to support $\RR_p$.
  
Now we can use our results from Section~\ref{sec:star}.
  Observe that $G[\{w_1, \dots, w_t, v\}]$ is a star centered at $v$.
  From Lemma~\ref{lem:star-kernel}, $R$ belongs to the kernel of $(\{R(w_1), \dots, R(w_t), R\}, \dP)$, where $R = R(v)$.
  Furthermore, from Lemma~\ref{lem:etoile-sans-centre}, 
  $(\{R(w_1), \dots, R(w_t)\}, \dP)$ has a planar support. Let $E_2$ be the edge set of this support (drawn in red in Figure~\ref{fig:dfs}).

  Thus, $G' := (\RR\setminus \{R\}, E_1 \cup E_2)$ is a support of $(\RR \setminus \{R\}, \dP)$. $G'$ is also planar because an edge in $E_1$ and that in $E_2$ can be drawn without crossing each other. As for every distinct $i,j \in [t]$, edges in $E(T(w_i))$ can be drawn in their private regions which are disjoint to that of $E(T(w_j))$.
  Finally, the edges in $E_2$ can be drawn outside the union of these private regions.

%
%
%
%
%
%
\end{proof}

\begin{figure}[h]
  \centering
  \begin{subfigure}{0.45\textwidth}
    \centering
    \includegraphics[width=0.9\textwidth]{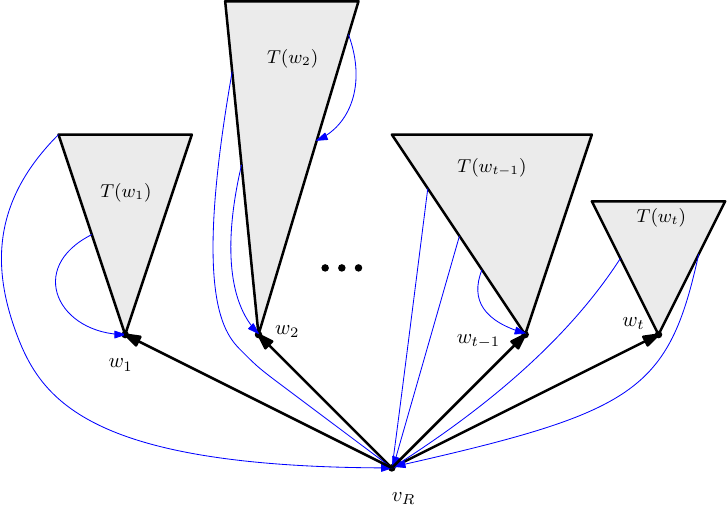}
  \end{subfigure}
  \begin{subfigure}{0.45\textwidth}
    \centering
    \includegraphics[width=0.9\textwidth]{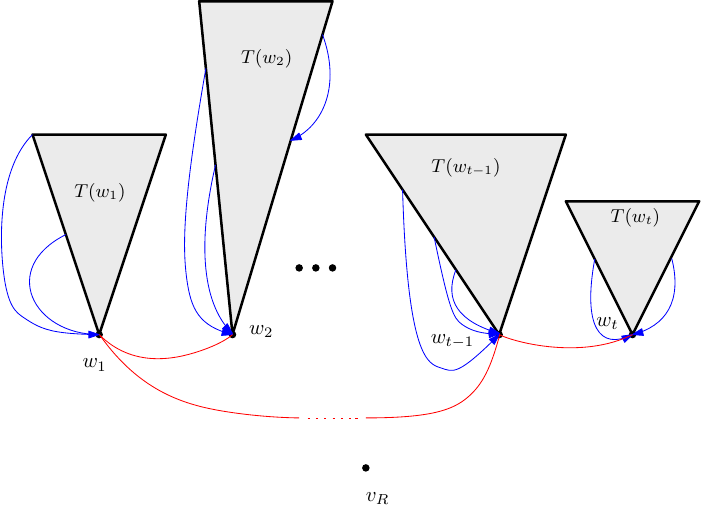}
  \end{subfigure}
  \caption{On the left, there is a schematic diagram of $G$, the support graph of $(\RR, \dP)$. On the right, there is a schematic diagram of $G'$, the support graph of $(\RR\setminus\{R\}, \dP)$. $G'$ is obtained from $G$ by shortcutting the cotree edges (drawn in blue), incoming to $v$, and adding the edges in $E_2$ (red edges among $w_j$'s) such that  $(\{R(w_1), \dots, R(w_t)\}, \dP)$ has a planar support.}
        \label{fig:dfs}
      \end{figure}

\section{Complete Family of Maximal Rectangles}
\label{sec:complete}

In this section, we state Lemma~\ref{thm:complete-rect} and give a proof sketch.
This lemma is used to prove Theorem~\ref{thm:main}.
Refer to Section~\ref{sec:complete-long} for the full proof of the lemma.
Recall that for a simple orthogonal polygon $P$, $\RR^K(P)$ is the set of all maximal rectangles with respect to $P$.

\begin{restatable}{lemma}{CompleteSupport}
  \label{thm:complete-rect}
  $(\RR^K(P),\dP)$ has a planar support.
\end{restatable}

\begin{hproof}
Given an orthogonal polygon $P$, we show that if the statement is true for a subpolygon $P' \subset P$, where $P\setminus P'$ is a rectangle, then it also true for $P$.
Thus, inductively one can argue that if $G_{P'}$ is a planar support for $(\RR^K(P'), \dP')$ then there exists a planar support $G_P$ for $(\RR^K(P), \dP)$. In particular, given $G_{P'}$ as an input, we give an algorithm to draw $G_P$.
\end{hproof}

\section{Planar Support Graphs}
\label{sec:main-results}


Finally we present our main result.

\BoundaryPlanarSupport*

\begin{proof}
  From Lemma~\ref{thm:complete-rect} it follows that $(\RR^K(P), \dP)$ has a planar support. By applying Lemma~\ref{lem:hereditary-planar} we prove that $(\RR, \dP)$ has a planar support, as $\RR \subseteq \RR^K(P)$.
\end{proof}




Additionally, we have two negative results, i.e., why local search fails for the Interior Cover problem and the Maximum Antirectangle problem for simple polygons.

\begin{theorem}
  \label{thm:interior}
    For every $r \ge 3$, there exists simple orthogonal polygon $P$ and an interior cover $\RR$ such that every minimal support graph of $(\RR, P)$ contains $K_{r,r}$.
  \end{theorem}

  \begin{proof}
    See Figure~\ref{fig:biclique}. For $r \ge 3$, one can construct a simple polygon $P$ and define a set of maximal rectangles as shown in Figure~\ref{fig:biclique} such that the minimal interior support graph is $K_{r,r}$.
  \end{proof}
  
  \begin{figure}
    \centering
  \begin{subfigure}[l]{0.45\textwidth}  
  \centering
        \includegraphics[width=\textwidth]{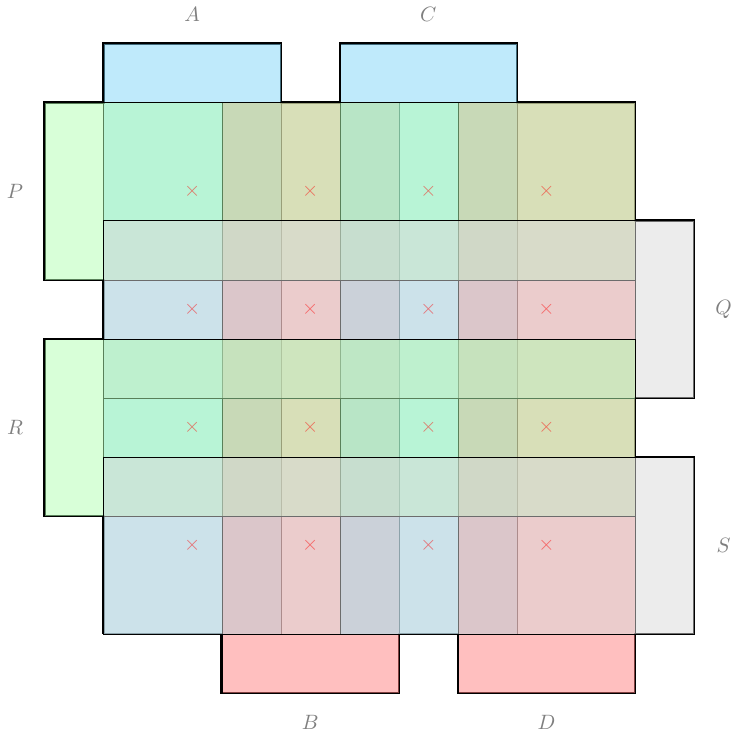}
  





      
      \caption{A simple polygon $P$ with \\
        $\RR = \{A, B, C, D, P, Q, R, S\}$.\\ The minimal support graph of $(\RR, P)$ is $K_{4,4}$.}
      \label{fig:biclique}
    \end{subfigure}
\begin{subfigure}[r]{0.49\textwidth}
        \centering
        \includegraphics[width=0.7\textwidth]{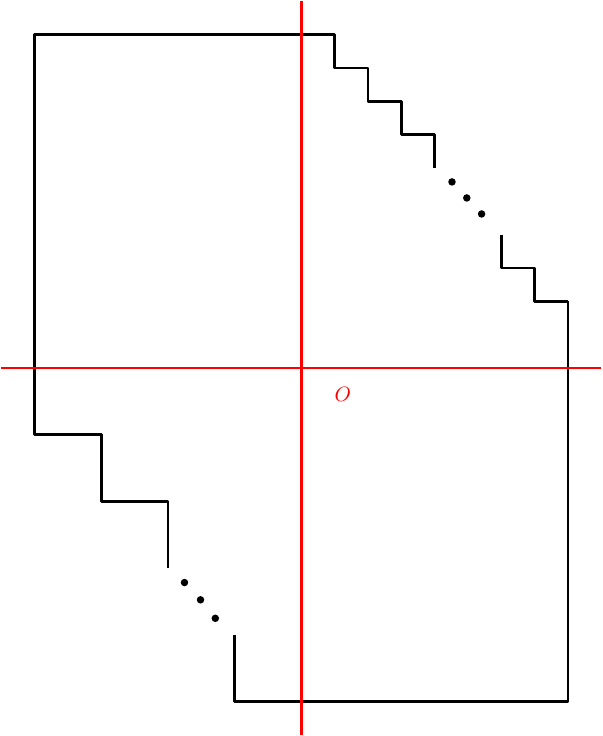}
        \caption{A schematic diagram of a polygon whose top-right quadrant has $r$ corners and its bottom-left quadrant has $s$ corners.}
        \label{fig:antirectangle}
      \end{subfigure}
      \caption{Instances that lead to (arbitrarily large) bicliques in the support graphs of the Interior Cover (left) and the Antirectangle problem (right).}
      \label{fig:neg-result}
\end{figure}
    Recall the visibility graph that we defined in Section~\ref{sec:visi-graph}.

    \begin{theorem}
      \label{thm:antirectangle}
      For every $r,s \ge 3$, there exists a simple polygon $P$ and a finite point set $Q \subset P$ such that $Q$ is a union of two antirectangles and the visibility graph defined on $Q$ is $K_{r,s}$.
    \end{theorem}

    \begin{proof}
      See Figure~\ref{fig:antirectangle}. Observe that the set of convex corners in the top-right (resp. bottom-left) quadrant forms an antirectangle.
      Also, observe that for every pair of corner, one from the top-right quadrant and the other from the bottom-left quadrant, there exist a rectangle containing them and in $P$. 
    \end{proof}

    An easy corollary of the theorem is the following.
    
    \begin{corollary}
      \label{cor:antirectangle}
      There exists simple polygons for which the local optimum solution can be arbitrarily smaller than the global optimum solution for the Maximum Antirectangle problem.
    \end{corollary}

    \begin{proof}
      Consider Figure~\ref{fig:antirectangle}. Assume $r = \Omega(n)$ and $s = O(\log n)$.
      The global optimum solution size is $r$, one such solution is picking all the convex corners of the polygon in the top-right quadrant.
      If the initial feasible solution for the local search algorithm is the set of corners in the bottom-left quadrant, then it can never improve its solution and that would be the local optimum solution.
    \end{proof}

\section{Conclusion}
\label{sec:open}

It is easy to see that $\theta_c \le \theta_b \le \theta$, where $\theta, \theta_b, \theta_c$ are the sizes of the minimum interior, boundary, and corner covers, respectively.
What is not known are the upper bounds. In particular, what is the smallest value of $\beta$ such that $\theta \le \beta\cdot\theta_b$? Thus, our PTAS for the boundary cover problem for simple polygons would imply a $(\beta + \epsilon)$-approximation for the corresponding interior cover problem.
Heinrich-Litan and L\"{u}bbecke \cite{H-LL2006JEA}, presented a number of computational results on the Polygon Covering problem. Notably, their experiments suggest $\theta \le 2 \theta_b$. An open question is to prove $\beta$ to be less than $2$, which would improve the $2$-approximation algorithm of Franzblau \cite{franzblau1989performance}. In Figure~\ref{fig:beta-lower-bound}, we present a lower bound instance that implies $\beta \ge 3/2$.

\begin{SCfigure}[1][!h]
  \includegraphics[width=0.35\textwidth]{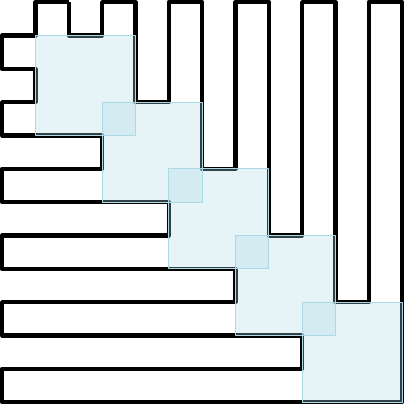}
  \caption{The boundary of the polygon is drawn with thick black lines. Here the rectangles in the minimum boundary cover are $6$ thin vertical rectangles and $6$ wide horizontal rectangles. Additionally for the minimum interior cover, we have $5$ blue squares. Thus, $\theta_b = 12$ and $\theta = 17$. This polygon can be generalized such that for every $k \ge 2$, $\theta_b = 2k +2$ and $\theta = 3k + 2$. Thus, $\beta \ge 3/2$.}
  
  \label{fig:beta-lower-bound}
\end{SCfigure}


  By far, the most challenging open question is whether there exists an $O(1)$-approximation algorithm running in polynomial time for the \emph{Interior Cover} problem for orthogonal polygons possibly with holes? 
  Can Erd\H{o}s' question (see Subsection~\ref{sec:visi-graph}) be answered in affirmative? Does there exist a primal-dual algorithm that could do so by returning interior cover and antirectangle solutions whose sizes are within a constant factor of each other?

  Possibly, an easier task would be to settle the status of the \emph{Corner Cover} problem for simple polygons. To the best of our knowledge we neither know a polynomial time exact algorithm or a proof of its $NP$-hardness.

%
\newpage

\bibliography{ref}

\newpage
 {\Huge{\bf Appendices}}
\appendix
\section{Missing Details from Section~\ref{sec:star}}
\label{sec:star-appendix}

\begin{lemma}
  \label{lem:corner}
   Let $N \in \RR_N$ such that it contains the top-right corner of $R_c$. If there exists $R \in \RR$ such that $R \cap N \cap \dP \neq \emptyset$, then $R$ is either aligned to the top side and/or the right side of $R_c$.
\end{lemma}

\begin{figure}
        \centering
        \includegraphics[width=0.5\textwidth]{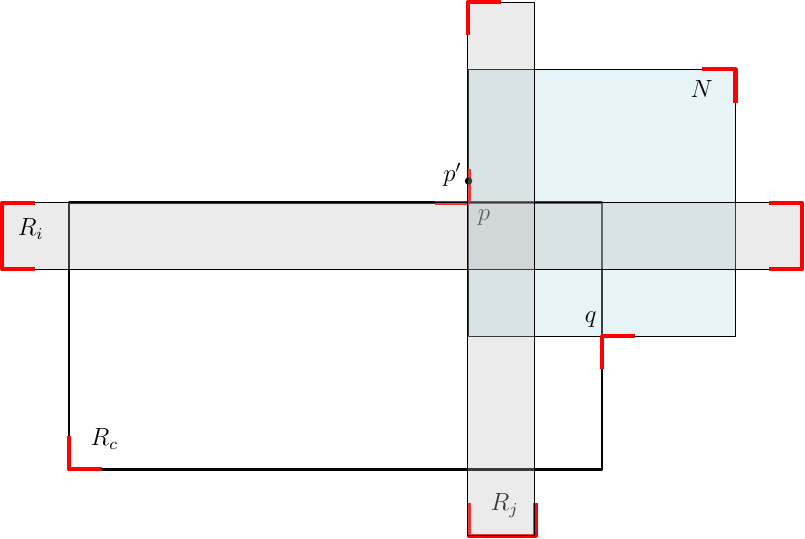}
        \caption{Let $\RR$ be a proper family with $R_c$ as the kernel rectangle. The above figure shows parts of $\dP$ in bold red. Corner intersection $N \cap R_c \cap \dP = \{p,q\}$. If $R_j$ belongs to $\RR$ then $p'$ witnesses the intersection $N \cap R_j \cap \dP$ but it is outside $R_c$, which implies $R_c$ is not a kernel. However, $R_i$ can be part of $\RR$ without contradicting the fact that $R_c$ is the kernel rectangle.}
        \label{fig:corner}
\end{figure}

  \begin{proof}
    Given $N$, a corner rectangle, there exists two points $p$ and $q$ where $\partial N \cap \partial R_c \neq \emptyset$. See Figure~\ref{fig:corner}.
    As $N \in \RR$ and $\RR$ is proper, at least one or both $p$ and $q$ belongs to $\dP$.
    If there exists $R \in \RR\setminus\{R_c\}$ such that $R \cap N \cap R_c \cap \dP \neq \emptyset$ then either $p$ or $q$ or both are contained in $R$. Without loss of generality, assume that $p \in R\cap \dP$. Then observe that the top side of $R$ is aligned with the top side of $R_c$. The left side of $R$ cannot be aligned with the left side of $N$, because otherwise there exists a point $p' \in \dP$ such that $p' \in (N \cap R \cap \dP)\setminus R_c$, which implies $R_c$ is not a kernel rectangle. See Figure~\ref{fig:corner}.
    Similarly, if $q \in R \cap \dP$ then the right side of $R$ is aligned with the right side of $R_c$.
  \end{proof}


  The following observation can be made from Lemma~\ref{lem:corner}.
\begin{observation}
  \label{lem:corner-disjoint}
    For $R_i, R_j \in \RR_N$, $R_i \cap R_j \cap \dP = \emptyset$.
  \end{observation}

  \begin{lemma}
    \label{lem:corner-pattern}
    Given a proper family $\RR$ with kernel rectangle $R_c$, there are at most $2$ corner rectangles, $N_i$ and $N_j$ containing a fixed corner of $R_c$. Furthermore, $N_i$ and $N_j$ pierce each other in a non-aligned way.
  \end{lemma}

 \begin{proof}
  Let $N_i$ and $N_j$ be two corner rectangles, i.e., $N_i, N_j \in \RR_N$ and they contain the top-right corner of $R_c$. We do a case analysis below and conclude that $N_i$ and $N_j$ can only have a piercing intersection in a non-aligned way.
    \begin{enumerate}
    \item
      Assume $N_i$ and $N_j$ have a corner intersection between themselves. Without loss of generality, further assume that the left side of $N_i$ is to the left of that of $N_j$. There are two subcases. First, the top side of $N_i$ is above that of $N_j$. Second, the top side of $N_i$ is below that of $N_j$.

      In the first subcase, the left side of $N_j$ is completely in the interior of $N_i$ and $R_c$, which implies that $N_j$ is not maximal. Hence, contradiction (see Figure~\ref{fig:corner-pattern}, part $a$). In the second subcase, $N_j \cap R_c \cap \dP = \emptyset$ which implies $\RR$ is not a proper family, hence a contradiction again (see Figure~\ref{fig:corner-pattern}, part $b$).

    \item
      Assume $N_i$ and $N_j$ have a piercing intersection. The first subcase is that they are aligned to each other and the second subcase is that they are not.
      
      Without loss of generality, assume one of the sides of $N_j$, $s_j$ is aligned to $s_i$, the corresponding side of $N_i$ and $s_j \subset s_i$. As $N_j$ is maximal, $s_j \cap \dP \neq \emptyset$. This also implies that $s_i \cap s_j \cap \dP \neq \emptyset$. Then it contradicts Observation~\ref{lem:corner-disjoint}
      (see Figure~\ref{fig:corner-pattern}, part $c$).

      For an example of $N_i$ and $N_j$ having a non-aligned piercing intersection, see part $d$ of Figure~\ref{fig:corner-pattern}.
      
    \end{enumerate}

    Next assume, that there is a third rectangle $N_k$ that also contains the top-right corner of $R_c$ and $N_i \prec N_k \prec N_j$. This implies $N_k \cap R_c \cap \dP = \emptyset$, which means $\RR$ is not proper and hence contradiction (see Figure~\ref{fig:corner-pattern}, part $d$).
      
  \end{proof}


  \begin{figure}[h]
        \centering
        \includegraphics[width=0.8\textwidth]{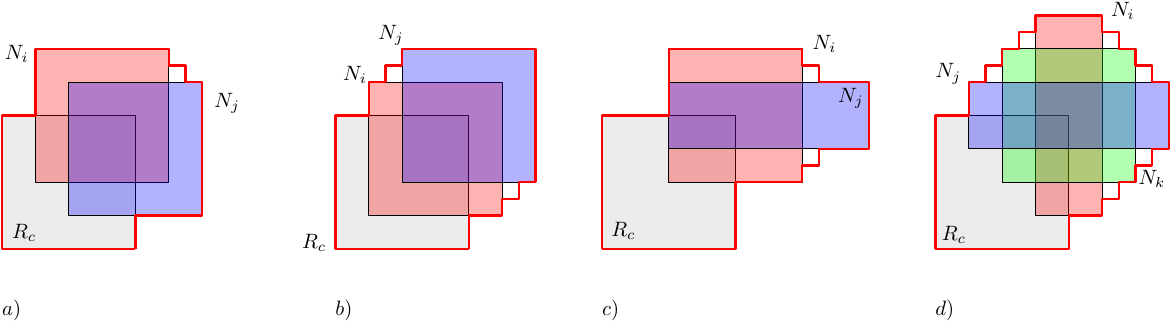}
        \caption{The subfigures are to support the various subcases of the proof of Lemma~\ref{lem:corner-pattern}. $\dP$ is drawn in thick red. The fill colors of rectangles $R_c, N_i, N_j, N_k$ are gray, red, blue, and green, respectively.}
        \label{fig:corner-pattern}
    \end{figure}

    As a consequence of Lemma~\ref{lem:corner-pattern}, we know that there are at most $8$ corner rectangles in $\RR$, at most $2$ per corner of $R_c$.
If there is a unique corner rectangle containing the top-right corner of $R_c$, we shall denote it by $N_{\nearrow}$. If there are two of them then we shall denote them by $N^v_\nearrow$ and $N^h_\nearrow$, where $N^v_\nearrow \prec N^h_\nearrow$. In general, a corner rectangle is denoted by $N^a_b$, where $a \in \{v,h,\emptyset\}$ and $b \in\{\swarrow, \nwarrow, \nearrow, \searrow\}$.

  \subsection{Vertical and Horizontal Rectangles}
  We define the subset of maximally vertical rectangles with respect to the partial order $\prec$ as $\RR^*_V := \max\{R \in \RR_V\}$. Similarly, we define the subset of minimally horizontal rectangles $\RR^*_H := \min\{R \in \RR_H\}$.

\begin{lemma}
    \label{lem:vertical-triple}
There does not exist $3$ rectangles $R_i, R_j, R_k \in \RR^*_V$ (resp. $\RR^*_H$) such that each one has a corner intersection with the other two in $\dP$, i.e., $R_i \cap R_j \cap \dP \neq \emptyset$, $R_j \cap R_k \cap \dP \neq \emptyset$, and $R_i \cap R_k \cap \dP \neq \emptyset$.
\end{lemma}

  
\begin{proof}
    Assume there exist $3$ rectangles, $R_i, R_j, R_k \in \RR$ such that $R_i \cap R_j \cap \dP \neq \emptyset$, $R_j \cap R_k \cap \dP \neq \emptyset$, and $R_i \cap R_k \cap \dP \neq \emptyset$.
    Combinatorially, triple of rectangles can be have pairwise corner intersection in two ways as shown in Figure~\ref{fig:vertical-triple}.
    If the triple intersect as shown on the left then $R_i \cap R_k \cap \dP = \emptyset$ which is a contradiction.
    If the triple intersect as shown on the right then the bottom side of $R_j$ and the right side of $R_i$ are completely in the interior of the union of the other two. This implies $R_i$ and $R_j$ are not maximal which is a contradiction.
  \end{proof}

 \begin{figure}[h]
        \centering
        \includegraphics[width=0.5\textwidth]{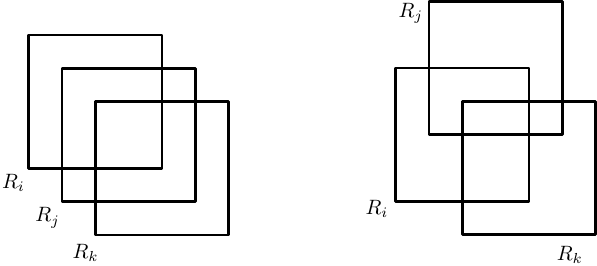}
        \caption{Three rectangles can intersect in two ways such that every pair has a corner intersection.}
        \label{fig:vertical-triple}
    \end{figure}


    Lemma~\ref{lem:vertical-triple} implies that we can linearly order the rectangles in $\RR^*_V$ (resp. $\RR^*_H$) as shown below.
    \begin{definition}
      \label{def:vert-order}
    Define a linear ordering $<_v$ on the rectangles in $\RR^*_V$, such that $R <_v R'$ if the left side of $R$ is to the left of that of $R'$, for every $R, R' \in \RR^*_V$.
  \end{definition}

  Let $\RR^*_V = \{V_1, \dots, V_k\}$, where $V_1 <_v \dots <_v V_k$.

  \begin{definition}
    \label{def:hori-order}
    Define a linear ordering $<_h$ on the rectangles in $\RR^*_H$, such that $R <_h R'$ if the bottom side of $R$ is below that of $R'$, for every $R, R' \in \RR^*_H$.
  \end{definition}

  Let $\RR^*_H = \{H_1, \dots, H_\ell\}$, where $H_1 <_h \dots <_h H_\ell$.

  \begin{lemma}
    \label{lem:2cor-1pie}
    Let $R_i, R_j \in \RR^*_V$ such that $R_i \cap R_j \cap \dP \neq \emptyset$. There does not exist $R_k \in \RR_V$ such that $R_k \prec R_i$ and $R_k \prec R_j$.
  \end{lemma}

   \begin{proof}
    Without loss of generality, assume $R_i <_v R_j$. For the sake of contradiction, lets assume such an $R_k$ exists. Then $R_k$ should be within the left side of $R_j$ and the right side of $R_i$ in order to have a piercing intersection with them. If $R_k$ is aligned to $R_i$ then they have an intersection with $\dP$ outside $R_c$, which means $R_c$ is not a kernel rectangle, hence a contradiction (see the left half of Figure~\ref{fig:vertical-consecutive}). Same is the case if it is aligned to $R_j$. If it is aligned to neither $R_i$ nor $R_j$, then $R_k \cap R_c \cap \dP = \emptyset$, which means $\RR$ is not proper, hence a contradiction again (see the right half of Figure~\ref{fig:vertical-consecutive}).
  \end{proof}

 \begin{figure}[h]
        \centering
        \includegraphics[width=0.65\textwidth]{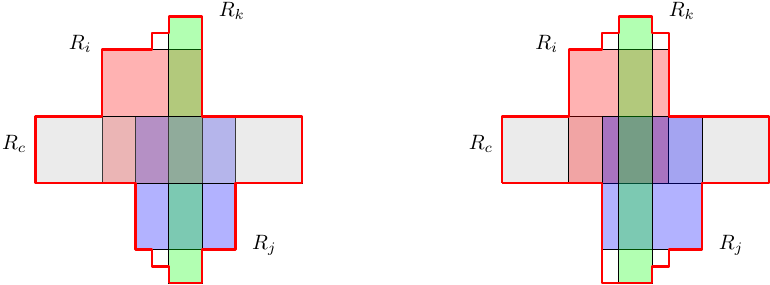}
        \caption{
          The above subfigures are to support Lemma~\ref{lem:2cor-1pie}.
        }
        \label{fig:vertical-consecutive}
    \end{figure}

    \begin{lemma}
      \label{lem:prec-star}
      Let $\RR^{\prec}_V(V_i) := \{R \in \RR_V \mid R \prec V_i\}$, for $V_i \in \RR^*_V$.
      There exists a support graph of $(\RR^{\prec}_V(V_i) \cup \{V_i\}, \dP)$ that is a star graph with $V_i$ as the center.
    \end{lemma}
    \begin{proof}
      Lemma~\ref{lem:2cor-1pie} implies that $V_i$ is a kernel rectangle of $(\RR^{\prec}_V(V_i) \cup \{V_i\}, \dP)$. Observation~\ref{lem:kernel-star} implies that the support graph of $(\RR^{\prec}_V(V_i) \cup \{V_i\}, \dP)$ is a star graph with $V_i$ as the center.
      \end{proof}
    Same arguments hold for $H_j \in \RR^*_H$.
    
  We refer $V_1, V_k, H_1,$ and $H_\ell$ to be terminal rectangles.

  If $V_1$ and $R_c$ have a common left blocker, then we call $V_1$ to be left-aligned. Likewise, if $V_k$ and $R_c$ have a common right blocker, then we call $V_k$ to be right-aligned. Conversely, if $H_1$ and $R_c$ have a common top blocker, then we call $H_1$ to be top-aligned. Finally, if $H_\ell$ and $R_c$ have a common bottom blocker, then we call $H_\ell$ to be bottom-aligned. Note that a terminal rectangle may or may not be aligned.

\subsection{Intersection between Corner and Non-corner Rectangles}

From Lemma~\ref{lem:corner} we know that any vertical or horizontal rectangle having an intersection with $N^a_\nearrow$ at $\dP$, should be top-aligned and/or right-aligned to $R_c$, where $a \in \{v,h,\emptyset\}$.

\begin{lemma}
  \label{lem:corner-non-corner}
  There are at most $2$ vertical rectangles in $\RR^*_V$ and $2$ horizontal rectangles in $\RR^*_H$ intersecting a corner rectangle at $\dP$. Among these up to $4$ rectangles, there are at most $1$ vertical terminal rectangle, and at most $1$ horizontal terminal rectangle. Furthermore, for a fixed corner of $R_c$, if there are two corner rectangles, then there are at most $1$ vertical and $1$ horizontal rectangle per corner rectangle.
\end{lemma}

  \begin{proof}
  Without loss of generality, fix the top-right corner of $R_c$. If there are corner rectangles, then either there is a single corner rectangle or there are two of them.

  For the first case, the corner rectangle is $N_\nearrow$ and there are possibly $4$ rectangles $V_i, V_k \in \RR^*_V$ and $H_j, H_\ell \in \RR^*_H$, where $1 \le i < k$ and $1 \le j < \ell$ such that they intersect $N_\nearrow$ at $\dP$. Recall that $V_k$ and $H_\ell$ are terminal rectangles (see Figure~\ref{fig:single-corner}).

  Similar is the case, if there are two corner rectangles, $N^v_\nearrow$ and $N^h_\nearrow$. $N^v_\nearrow$ has an intersection with $V_i$ and $H_\ell$ at a boundary point $p$, whereas, $N^h_\nearrow$ has an intersection with $V_k$ and $H_j$ at another boundary point $q$ (see Figure~\ref{fig:double-corner}).

\end{proof}

See Figure~\ref{fig:all-corners} for an illustration of all corners having two corner rectangles each, and also Figure~\ref{fig:graph-all-corners} for the corresponding support graph.

\begin{figure}[h]
  \centering
\begin{subfigure}[l]{0.49\textwidth}
        \centering
        \includegraphics[width=0.9\textwidth]{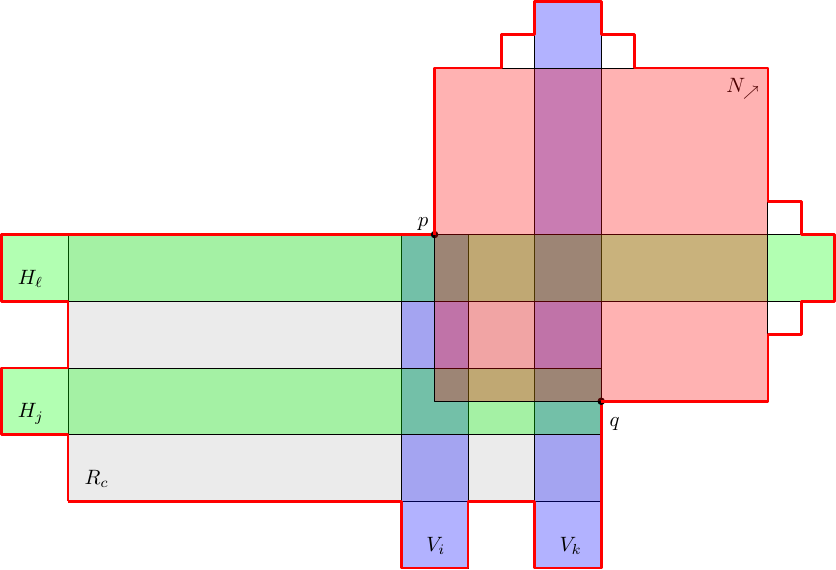}
        \caption{
          $N_\nearrow$ is the unique corner rectangle containing the top-right corner of $R_c$. Boundary points $p \in V_i \cap H_\ell \cap N_\nearrow \cap \dP$ and $q \in V_k \cap H_j \cap N_\nearrow \cap \dP$.
        }
        \label{fig:single-corner}
      \end{subfigure}
      \hfill
\begin{subfigure}[r]{0.49\textwidth}
        \centering
        \includegraphics[width=0.9\textwidth]{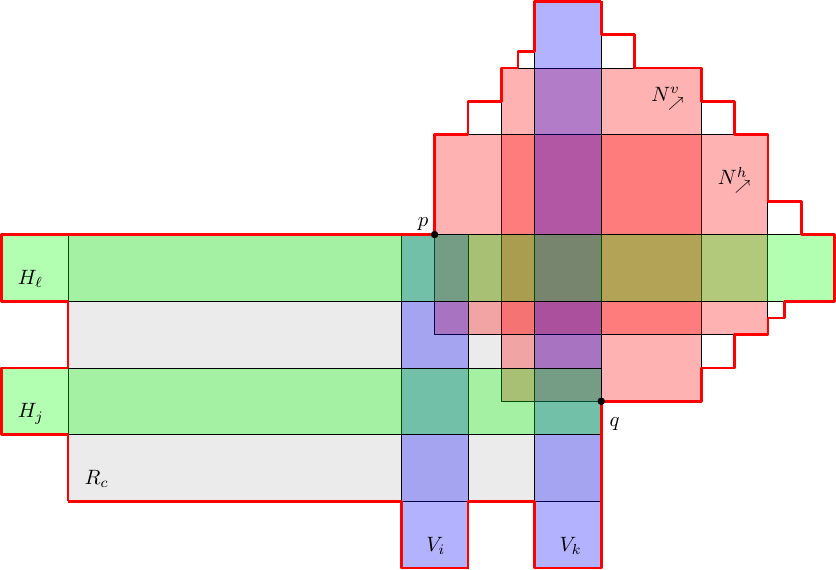}
        \caption{
          $N^v_\nearrow$ and $N^h_\nearrow$ are the two corner rectangles containing the top-right corner of $R_c$.
          Boundary points $p \in V_i \cap H_\ell \cap N^h_\nearrow \cap \dP$ and $q \in V_k \cap H_j \cap N^v_\nearrow \cap \dP$.
        }
        \label{fig:double-corner}
      \end{subfigure}
      \caption{The intersection patterns between corner and non-corner rectangles.}
      \label{fig:single-double-corner}
    \end{figure}


We make the following observation.

\begin{lemma}
  If there are two corner rectangles $N^v_b, N^h_b$ for a fixed corner $b$ of $R_c$, where $b = \{\nearrow, \nwarrow, \searrow, \swarrow\}$, then there exists a support graph such that degree of both these corner rectangles in this graph is at most $2$.
\end{lemma}

\begin{proof}
  Fix $b = \nearrow$ and consider $N_\nearrow^v$. From Lemma~\ref{lem:corner-non-corner}, and Figures~\ref{fig:double-corner} and~\ref{fig:all-corners}, we know there exists at most one maximally horizontal rectangle $H_j$, and similarly at most one maximally vertical rectangle $V_k$, such that there is a unique point $p \in N_\nearrow^v \cap H_j \cap V_k \cap \dP \neq \emptyset$.
  Recall, from Lemma~\ref{lem:prec-star}, there exists a support graph such that non-maximally horizontal (resp. vertical) rectangles are adjacent to a unique maximally horizontal (resp. vertical) rectangle.
  Thus, in a minimal support graph, $N_\nearrow^v$ has degree at most two.
  Similar arguments hold for $N_\nearrow^h$, and also for other values of $b$.
\end{proof}




\begin{remark}
  For a given corner $b \in \{\nearrow, \nwarrow, \searrow, \swarrow\}$, we shall ignore all pair of corner rectangles $N^v_b$ and $N^h_b$, if they exist, from our following graph drawing algorithm. Suppose $v$ is a degree-$2$ vertex, incident to edges $(u,v)$ and $(v,w)$. Instead of drawing the pair of edges, one can draw an edge $(u,w)$ and place $v$ anywhere in the interior of the drawing of this edge. On the other hand, if $v$ is a degree-$1$ vertex and $(u,v)$ is the only edge incident to it, then this edge can be drawn arbitrarily close to $u$ without intersecting any other edges.
\end{remark}

\begin{figure}[h]
        \centering
        \includegraphics[width=0.65\textwidth]{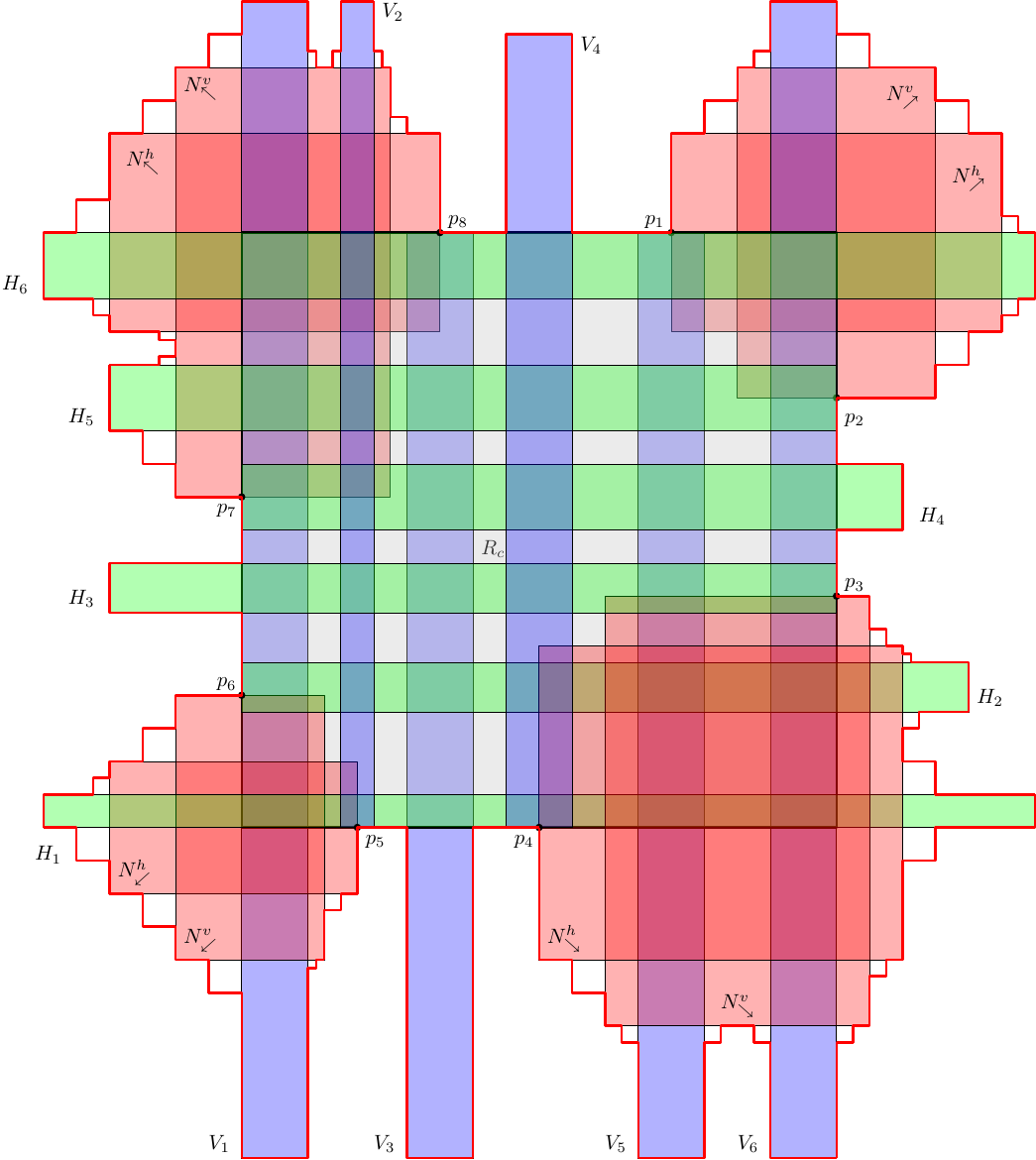}
        \caption{
          This figure is an extension of Figure~\ref{fig:double-corner} to all the $4$ corners of $R_c$. See Figure~\ref{fig:graph-all-corners} for the corresponding support graph.
        }
        \label{fig:all-corners}
\end{figure}

\begin{figure}[h]
  \centering \includegraphics[width=0.65\textwidth]{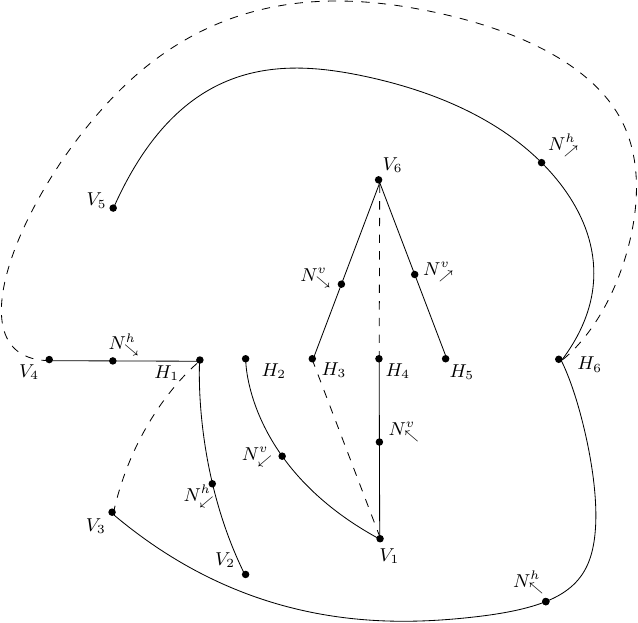}
        \caption{
          This figure shows the drawing of the support graph of the arrangement shown in Figure~\ref{fig:all-corners}.
          The edges drawn in solid lines are to support rectangles intersecting at $p_i$, $1 \le i \le 8$. Whereas the edges that do not involve corner rectangles are drawn in dashed lines.
        }
        \label{fig:graph-all-corners}
\end{figure}


\begin{figure}[h]
        \centering
        \includegraphics[width=0.65\textwidth]{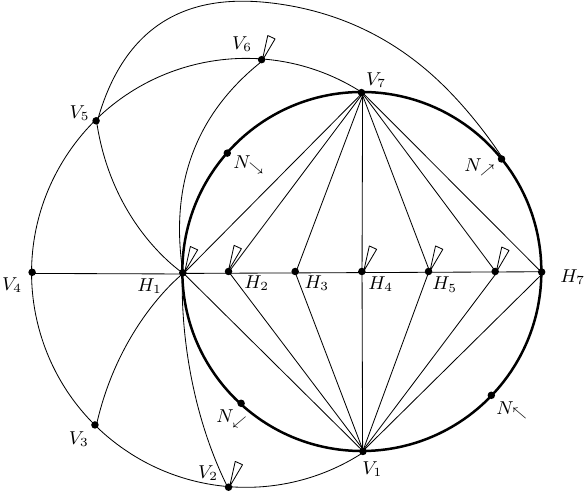}
        \caption{Schema of a planar support graph. $\RR^*_V = \{V_1,\dots, V_7\}$, $\RR^*_H = \{H_1, \dots, H_7\}$, $\RR_N = \{N_\swarrow, N_\nwarrow, N_\nearrow, N_\searrow\}$.
          In comparison to Figure~\ref{fig:kernel-schema}, here $H_7$ is not top-aligned and $V_5 \cap N_\nearrow \cap \dP \neq \emptyset$, hence there is the edge $(V_5, N_\nearrow)$.}
        \label{fig:kernel-corner}
      \end{figure}

\subsection{Algorithm to Draw Planar Support for $\RR\setminus ker(\RR)$}
\label{sec:algo-kernel-less}
This section is an extension of Section~\ref{sec:star}.
  In the following we give a simple algorithm to draw a planar support of $(\RR\setminus ker(\RR), \dP)$.
  We put an edge between the left-aligned (resp.~right-aligned) rectangle, if it exists, and every horizontal rectangle. Similarly, we put an edge between the top-aligned (resp.~bottom-aligned) rectangle, if it exists, and every vertical rectangle. Without loss of generality, assume that all the aligned rectangles exist.
  \subsubsection{Placing the Vertices}
  Draw a circle and place the terminal vertices $V_1, H_1, V_k, H_\ell$ on it in this order i.e., alternating between the horizontal and vertical rectangles as shown in Figure~\ref{fig:kernel-schema} where the circle is drawn in bold.
  Place the vertices corresponding to the corner rectangle between the terminal rectangles, i.e.,
  $N_\swarrow$ between $V_1$ and $H_1$,
  $N_\searrow$ between $V_k$ and $H_1$,
  $N_\nearrow$ between $V_k$ and $H_\ell$, and
  $N_\nwarrow$ between $V_1$ and $H_\ell$.
  
  Define a line joining $H_1$ and $H_\ell$ in the inner face of the cycle. Place the vertices corresponding to the other horizontal rectangles $H_j \in \RR^*_H\setminus \{H_1, H_\ell\}$, on this line obeying the linear order defined in Definition~\ref{def:vert-order}.
  Similarly, define a line joining $V_1$ and $V_k$ in the outer face of the cycle.
  Place the vertices corresponding to the other vertical rectangles $V_i \in \RR^*_V\setminus\{V_1, V_k\}$, on this line obeying the linear order defined in Definition~\ref{def:hori-order}.
  
  For every $V_i \in \RR^*_V$, if there exists $R$ such that $R \prec V_i$, then define a \emph{private} region in the plane (shown by a triangle with a vertex at $V_i$ in Figure~\ref{fig:kernel-schema}). Place the vertices corresponding to all such rectangles in the $V_i$'s private region.
  Similarly define private regions for $H_j \in \RR^*_H$, if required, and place all vertices corresponding to rectangles $R$ in the region, where $H_j \prec R$.

  \subsubsection{Drawing the Edges}
  \begin{enumerate}
    \item {\bf $(\RR^*_V \times \RR^*_H)$.}
  If $V_1$ is left-aligned with $R_c$ then draw an edge $(V_1,H_j)$, for every $H_j \in \RR^*_H$ as shown in Figure~\ref{fig:kernel-schema}. Similarly, if $V_k$ is right-aligned then draw an edge $(V_k,H_j)$, for every $H_j \in \RR^*_H$.
  Conversely, if $H_1$ is bottom-aligned then draw an edge $(V_i, H_1)$, for every $V_i \in \RR^*_V$. Likewise, if $H_\ell$ is top-aligned then draw an edge $(V_i, H_\ell)$, for every $V_i \in \RR^*_V$.

\item {\bf $(\RR_V \times \RR_V)$.}
    If $V_i \cap V_{i+1} \cap \dP \neq \emptyset$ then draw an edge $(V_i,V_{i+1})$ for every $V_i \in \RR^*_V$ where $1 \le i < k$, along the line where the $V_i$'s are placed, as shown in Figure~\ref{fig:kernel-schema}.
  Draw an edge $(R,V_i)$ inside $V_i$'s private region, for every $V_i \in \RR^*_V$ and $R \prec V_i$.

\item {\bf $(\RR_H \times \RR_H)$.}
  If $H_j \cap H_{j+1} \cap \dP \neq \emptyset$ then draw an edge $(H_j,H_{j+1})$ for every $H_i \in \RR^*_H$ where $1 \le j < \ell$, along the line where the $H_j$'s are placed, as shown in Figure~\ref{fig:kernel-schema}.
  Draw an edge $(R,H_j)$ inside $H_j$'s private region, for every $H_j \in \RR^*_H$ and $H_j \prec R$.

\item {\bf $(\RR_N \times (\RR^*_V \cup \RR^*_H))$.}
  Lemma~\ref{lem:corner} implies that if $N_\nearrow \cap R \cap \dP \neq \emptyset$ where $R \in \RR^*_V \cup \RR^*_H$ then $R$ is either top-aligned and/or right-aligned to $R_c$. If $R$ is top-aligned to $R_c$ and $R \in \RR^*_H$ then $R = H_\ell$. Draw the edge $(N_\nearrow, H_\ell)$.
  If there exists $R' \in \RR^*_V$ such that $R' \cap N_\nearrow \cap \dP \neq \emptyset$ and $R'$ is top-aligned with $R_c$ then draw the edge $(N_\nearrow, R')$ if $N_\nearrow \cap H_\ell \cap \dP = \emptyset$, otherwise do not draw the edge $(N_\nearrow, R')$. See Figure~\ref{fig:kernel-corner}. The case where $R$ is top-aligned is identical and we accordingly draw the edges. Finally, we iterate over the other corner rectangles in $\RR_N$ and do a similar case analysis.
  \end{enumerate}
      
\EtoileSansCentre*
\begin{proof}

  It is easy to see with the help of Figures~\ref{fig:kernel-schema} and~\ref{fig:kernel-corner} that the resultant graph $G$ is planar.
  We claim that $G$ is also a support graph. For $p \in \dP$ consider $\RR_p \setminus \{R_c\}$.
  Consider $R, R' \in \RR_p \setminus \{R_c\}$. We claim $R$ and $R'$ are connected in $G$.
  \begin{enumerate}
    \item
  If $R, R' \in \RR_V$ then let $V, V' \in \RR^*_V$ such that $R \preceq V$ (i.e., $R \prec V$ or $R=V$) and $R' \preceq V'$. From Lemma~\ref{lem:prec-star}, $(R,V)$ and $(R',V')$ are edges in $G$. This also means $V$ and $V'$ must be consecutive in the ordering $<_v$ and hence they are adjacent in $G$. Thus, $R$ and $R'$ are connected in $G$. 
  Same is the case for $R, R' \in \RR_H$.
  
\item
  If $R \in \RR_V$ and $R' \in \RR_H$. Let $V \in \RR^*_V$ and $H \in \RR^*_H$ such that $R \prec V$ and $H \prec R'$. Then observe that either $V$ is a left/right aligned terminal rectangle or $H$ is a top/bottom aligned terminal rectangle, or both. In either cases, $(H,V)$ is an edge in $G$. Hence, $R$ and $R'$ are connected in $G$.

\item
  If $R \in \RR_N, R' \in \RR_V$. Without loss of generality, let $R = N_\nearrow$ and let $V \in \RR^*_V$ and $R \preceq V$. From Lemma~\ref{lem:corner} there are two cases, either $V$ is top-aligned with $R_c$ or right-aligned with $R_c$. In either cases, there is an edge drawn between $V$ and $N_\nearrow$. For the former case, see Figure~\ref{fig:kernel-corner}. For the latter case, see Figure~\ref{fig:kernel-schema}. Thus, $R$ and $N_\nearrow$ are connected. The other cases involving $\RR_N$ and $\RR_H$ are identical.
\end{enumerate}
\end{proof}

\section{Missing Details from Section~\ref{sec:complete}}
\label{sec:complete-long}

Given a simple orthogonal polygon $P$, recall $\RR^K(P)$ is the set of all maximal rectangles with respect to $P$.
In this section we give the full proof of Lemma~\ref{thm:complete-rect}, i.e., $(\RR^K(P), \dP)$ has a planar support, whose proof sketch can be found in Section~\ref{sec:complete}. To this end we shall present a number of definitions and lemmata.


\begin{definition}
  Given a rectangle $X$, we denote its top, bottom, left and right sides by $t_X, b_X, \ell_X,$ and $r_X$, respectively.
\end{definition}

\begin{definition}[Vertically Blocked Rectangles]
  Given an orthogonal polygon $P$, a maximal rectangle $R \subseteq P$ is called vertically blocked if $R$ has a top blocker $(x,y_1)$ and a bottom blocker $(x,y_2)$, and
  the line joining these two blockers is called the vertical blocker of $R$.
  We say $R$ has two distinct vertical blockers $\ell_1$ and $\ell_2$ if $\ell_1$ joins top blocker $(x_1,y_1)$ and bottom blocker $(x_1,y_2)$ of $R$, 
  $\ell_2$ joins top blocker $(x_2,y_1)$ and bottom blocker $(x_2,y_2)$, and
  there exists $x$ where $x_1 < x < x_2$, such that either $(x,y_1) \notin \dP$ or $(x,y_2) \notin \dP$, or both.   
\end{definition}

\begin{lemma}
  \label{lem:vert-blocker}
  Let $R \in \RR^K(P)$ has two distinct vertical blockers $\ell_1$ and $\ell_2$. Let $\RR(\ell_1, \ell_2)$  be the subset of rectangles in $\RR^K(P)$ that are contained in the vertical slab defined by $\ell_1$ and $\ell_2$.
  Then in any minimal support $G$ of $(\RR^K(P), \dP)$, $R$ separates the two sets $\RR(\ell_1,\ell_2)$ and  $\RR^K(P)\setminus (\RR(\ell_1,\ell_2) \cup \{R\})$.
\end{lemma}

\begin{proof}
This follows because if there exists $p \in \dP$ such that $ p \in R_i \cap R_j$, where $R_i \in \RR(\ell_1,\ell_2)$ and $R_j \in \RR^K(P)\setminus (\RR(\ell_1,\ell_2) \cup \{R\})$, then $p \in R$.
\end{proof}

\begin{lemma}
  {\label{lem:opp-blockers}}
  Given a simple orthogonal polygon, no $2$ distinct maximal rectangles have the same blockers in the opposite directions, viz., top and bottom or left and right.
\end{lemma}

\begin{proof}
  For the sake of contradiction, assume that there exists a simple orthogonal polygon $P$, and $2$ distinct maximal rectangles $A$ and $B$ that have the same blockers on their top and bottom direction. Consider the left sides of $A$ and $B$.
  Either their left sides are distinct or same. If they are distinct then one of the left side is to the left side of the other rectangle.  Let $A$'s left side be to the left of $B$'s left side. This means $B$'s left blocker is contained in $A$, which contradicts the fact that $A$ is a maximal rectangle.
  The other case is when their left sides are same. Then one argues the same for their right sides. If their right sides are distinct then one can argue similarly as above and lead to a contradiction. If the right sides are also same then the $A$ and $B$ are the same rectangle.
\end{proof}

\begin{lemma}
Given a simple orthogonal polygon, if $2$ maximal rectangles have the set of their common blockers exactly in $2$ adjacent directions, viz., top and left, left and bottom, bottom and right, or right and top. Then the two rectangles are distinct.
\end{lemma}

\begin{proof}
  Suppose the two rectangles are not distinct. Then the set of common blockers is in all the $4$ directions.
\end{proof}

\begin{definition}[Extension of a Rectangle]
We say $R$ is a top (resp. bottom, left, right) extension of $R'$ with respect to a polygon $P$ if $R$ can be obtained by scaling $R'$ in the top (viz., bottom, left, right) direction and $R$ is maximal with respect to $P$.
\end{definition}

\begin{lemma}[Forest Support on Segments]
  \label{lem:laminar}
  Given a simple orthogonal polygon $P$, and a set of maximal rectangles $\RR$ such that for every $R \in \RR$, their bottom (resp. top, left, and right) blocker exists on a segment $s \in \dP$. Then there exists a support graph on the hypergraph $(\RR,s)$ that is a forest.
\end{lemma}

\begin{proof}
  The horizontal projections of the rectangles in $\RR$ forms a set of intervals. Moreover, we claim that the set of intervals forms a laminar family i.e., if two such distinct intervals $R_X:= [a,b]$ and $R_x':= [c,d]$ intersect then either $R_x \subset R_x'$ or  $R_x' \subset R_x$. Suppose otherwise is true, and without loss of generality, we assume, $a < c < b < d$. Then left blocker of $\RR_x'$ is contained in $\RR_x$, which is a contradiction.

  For every laminar family of intervals one can define a forest, where we have a node for every interval and $R_x'$ is a descendant of $R_x$ if $R_x' \subset R_x$.
\end{proof}

Moreover, observe that the forest support graph of $(\RR,s)$ is a tree, if the intersection graph of the intervals that are obtained by projecting the rectangles in $\RR$ on $s$ is connected.


We borrow some of the definitions and observations made by Gy{\H{o}}ri and Mezei \cite{gyHori2019mobile}.

\begin{definition}[Horizontal R-Tree \cite{gyHori2019mobile}]
  Given a simple orthogonal polygon $P$, we horizontally cut $P$ along all its concave corners. This partitions $P$ into a set of interior disjoint rectangles $R_1,\uplus\dots,\uplus R_m$. We define a graph $T(P) = (V_T,E_T)$ where we have a vertex $v_i$ corresponding to $R_i$, for every $1\le i \le m$, and $(v_i,v_j) \in E_T$ if $\partial R_i \cap \partial R_j \neq \emptyset$.
\end{definition}

\begin{lemma}[\cite{gyHori2019mobile}]
Given a simple orthogonal polygon $P$, $T(P)$ is a tree.
\end{lemma}

We make the following assumption for the ease of exposition without losing generality as done in \cite{gyHori2019mobile}.
We assume that given a simple orthogonal polygon $P$ every maximal rectangle in $P$ has positive horizontal and vertical lengths. If $P$ is such that there exists a rectangle with zero horizontal/vertical length, we can always do transformation preserving the number of maximal rectangles in $P$.

\begin{lemma}
  \label{lem:struct}
  Let $P$ be a simple orthogonal polygon with $B$ be a leaf of $T(P)$. Furthermore, let $R, R' \in \RR^K(P) \setminus \RR_B$ and $G_P$ be a minimal support graph of $(\RR^K(P), \dP)$. If $(R,R') \in E(G_{P})$ then we call $(R,R')$ an auxiliary edge.
  If we keep contracting all the auxiliary edges then eventually the contracted vertex has a unique neighbor $R''$ such that $R'' \in \RR_B$.
\end{lemma}

\begin{proof}
For the sake of contradiction, assume that there exists a contracted vertex that has multiple neighbors in $\RR_B$.
This means there exists $R \in \RR^K(P') \setminus \RR_B$, such that $R$ is biconnected to some $R'' \in \RR_B$ in $G_P'[\RR^K(P') \setminus \RR_B]$. This in turn implies there is a cycle in $T(P)$, which is a contradiction.
\end{proof}




\CompleteSupport*

 \begin{figure}[h]
        \centering
        \includegraphics[width=0.75\textwidth]{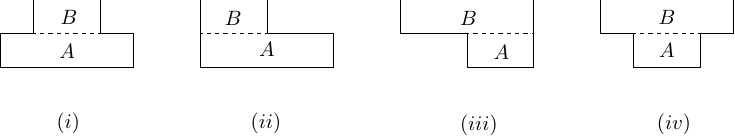}
        \caption{The possible ways rectangles $A$ and $B$ can exist in the polygon $P$. The two figures $(i)$ and $(ii)$ have $B_x \subset A_x$. Whereas, figures $(iii)$ and $(iv)$ have $A_x \subset B_x$.}
        \label{fig:tree-pattern-AB}
    \end{figure}

\begin{proof}
In order to prove the statement inductively, we prove something stronger, i.e., given $(\RR^K(P), \dP)$, there is a plane graph of the support such that the vertices corresponding to the vertically-blocked rectangles of $P$ are in its outer face.

We consider $T(P)$, the horizontal R-Tree of $P$ and prove the statement inductively on the number of vertices of $T(P)$.
  The base case is when there is a single vertex in $T(P)$, i.e., $P$ is a rectangle and the claim of the lemma is trivially true.

  Let $A \subset P$ be rectangle that corresponds to a leaf of $T(P)$.
  Let $P' := P \setminus A$. Also let $B$ be the parent of $A$ in $T(P)$.
  See Figure~\ref{fig:tree-pattern-AB}, for the possible configurations of $A$ and $B$.
  The inductive case is that, if the statement is true for $P'$ then it is also true for $P$.
Let $G_{P'}$ be the planar support graph that we get from the inductive hypothesis.

  As, the vertices of the support graphs are rectangles, we shall interchangeably use rectangles and vertices to denote the same, which would be clear from the context.

  Let $\RR_A$ be the subset of maximal rectangles in $\RR^K(P)$ that have their bottom blockers in $b_A$. Likewise, let $\RR_B$ be the set of maximal rectangles in $\RR^K(P')$ that have their bottom blockers in $b_B$.
  The high-level idea is to attach the vertices of $\RR_A$ to some of the vertices of $\RR_{B}$ and also remove some of the subtrees at $\RR_B$ and attach it to vertices in $\RR_A$. 

  Although a blocker of a maximal rectangle is any point in $\dP$ that intersects the rectangle but in the following we will relax this definition to include sides of $\dP$ that intersects the rectangle.
  
  As every rectangle $R \in\RR_A$ has the bottom side of $A$ as its bottom blocker, therefore from Lemma~\ref{lem:opp-blockers}, $R$ can be uniquely determined by its top blocker $t_R$. Similarly, every rectangle in $\RR_B$ can be uniquely determined by its top blocker.

  We define a function $f : \RR_A \setminus \{A\} \rightarrow \RR_B$, where for every $R \in \RR_A \setminus \{A\}$ there exists $f(R) \in \RR_B$ and $t \in \dP'$ such that $t$ is a top blocker of both $R$ and $f(R)$. We say $f(R)$ is the image of $R$.
  Also, observe that Lemma~\ref{lem:opp-blockers} implies that $f$ is an injective mapping.





  There are two cases: given $R \in \RR_A \setminus \{A\}$ either $f(R)\in \RR^K(P)$ or not.
  We say $R \in \RR_A\setminus \{A\}$ is a Type-1 rectangle if $f(R) \in \RR^K(P)$, otherwise $R$ is a Type-2 rectangle. If $A \in \RR_A$, then $A$ is neither Type-1 nor Type-2 rectangle. Henceforth, we denote the set of Type-1 (resp. Type-2) rectangles as $\RR^{(1)}_A$ (resp. $\RR^{(2)}_A$).

  Observe that Type-1 rectangles either have $\ell_A$ as the unique left blocker or have $r_A$ as the unique right blocker, or both. The rectangle that has both $\ell_A$ as the unique left blocker and $r_A$ as the unique right blocker is $A^\uparrow$, the top extension of $A$ in $P$, and we call it the root of the rectangles in $\TT(\RR_A)$.
  Whereas, Type-2 rectangles have both left and right blockers other than sides $\ell_A$ and $r_A$. Thus, for every $R \in \RR^{(2)}_A$, $R = f(R)^{\downarrow}$, the bottom extension of $f(R)$ in $P$.
  
  For $R \in \RR^K(P)\setminus \RR_A$, $R$ is said to be \emph{red} if either $R \cap R' \cap \dP \neq \emptyset$, for $R' \in \RR_A$, or for every $R''$, such that $R'' \cap R \cap \dP \neq \emptyset$, $R''$ is also red. See Figures~\ref{fig:type-12}, \ref{fig:type-12-fP}, and~\ref{fig:graph-type-12}, where $X$ and $Z$ are red because of the first case, whereas, $Y$ is red because of the second case.
  
  Hereafter the red vertices are ignored from the analysis for the sake of brevity.

  \begin{lemma}
    \label{lem:path-supp}
There exists support graphs of $(\RR_A^{(1)}, b_A)$ that are path graphs, where $\RR_A^{(1)}$ is the set of Type-1 vertices in $\RR_A$.
  \end{lemma}

  \begin{proof}
    If $\RR^{(1)}_A \neq \emptyset$ then it can be grouped into two subsets where one subset $L(\RR^{(1)}_A)$ contains the Type-1 rectangles whose unique left blocker is $l_A$ and the other subset $R(\RR^{(1)}_A)$ contains those whose unique right blocker is $r_A$. There exists $A^\uparrow$, the top extension of $A$, in both these subsets. Thus, we can order the rectangles from left to right, such that if $R_i \in L(\RR^{(1)}_A)$ and $R_j \in R(\RR^{(1)}_A)$ then $R_i <_1 R_j$, with the top extension of $A$ in between. For $R_i, R_j \in L(\RR^{(1)}_A)$, $R_i <_1 R_j$ if $r_{R_i} < r_{R_j}$.
    Thus the one whose right side is closest to $\ell_A$ appears first in this ordering $<_1$.
    Conversely, $R_i, R_j \in R(\RR^{(1)}_A)$, $R_i <_1 R_j$ if $l_{R_i} < l_{R_j}$. Thus the one whose left side is closest to $r_A$ appears last in the ordering $<_1$. See Figures~\ref{fig:type-12} and~\ref{fig:graph-type-12}.
    
    We define the path $\PP$ on $\RR^{(1)}_A$ such that $R_i$ and $R_{i+1}$ are adjacent in $\PP$ if they are consecutive with respect to the linear ordering $<_1$. From Lemma~\ref{lem:laminar}, it follows that $\PP$ is a support graph of $(\RR_A^{(1)}, b_A)$.
  \end{proof}

  We run Algorithm~\ref{alg:draw} on $G_{P'}$ and it outputs $G_P$.
  From Lemmata~\ref{lem:comp-support},~\ref{lem:comp-planar}, and~\ref{lem:comp-outer}, the proof follows.
\end{proof}


  
  \hide{
  \begin{lemma}
    If $A_x \subset B_x$ then there exists a tree support graph $G$ of $(\RR_A, b_A)$ and tree support graph $H$ of $(\RR_B, b_B)$ such that $G$ is isomorphic to a subgraph of $H$.
  \end{lemma}

  \begin{lemma}
    If $B_x \subset A_x$ then there exists a tree support graph $G$ of $(\RR_A, b_A)$ and tree support graph $H$ of $(\RR_B, b_B)$ such that $H$ is isomorphic to a subgraph of $G$.
  \end{lemma}

  \begin{proof}
Observe that every vertex in $\RR_A$ except $v_A$, has an image in $\RR_B$. Moreover, if there exists $R_i, R_j \in \RR_A$ such that $R_i <_p R_j$ then  $R_i' <_p R_j'$ where $R_i', R_j' \in \RR_B$ are images of $R_i$ and $R_j$, respectively. Thus, the tree structure defined on $\RR_B$ due to the partial order $<_p$ is isomorphic to the tree structure defined on $\RR_A \setminus \{v_A\}$.
  \end{proof}
  
  \begin{lemma}
    Given $B_x \subset A_x$, if there exists a planar support $G_{P'}$ for $(\RR^K(P'), \dP')$ then there exists a planar support $G_P$ for $(\RR^K(P), \dP)$.
  \end{lemma}
  \begin{proof}
Consider $G_{P'}$ and attach vertex $v_A$ by adding the edge $(v_A, v_B)$. We also relabel the vertices of $\RR_B$, for every $R' \in \RR_B$, we relabel it to $R \in \RR_A$. The resultant graph is also planar as $v_A$ has degree $1$. 
  \end{proof}

} 

  
  \subsubsection*{Algorithm producing $G_{P}$ from $G_{P'}$.}
  
  Observe, $\RR^K(P) = (\RR^K(P') \cup \RR_A)\setminus f(\RR^{(2)}_A)$, where $f(\RR^{(2)}_A) := \{f(R) \in \RR_B \mid R \in \RR^{(2)}_A \}$.
See Figures~\ref{fig:type-12} and~\ref{fig:graph-type-12} for an example.
  \begin{algorithm}
    \caption{Algorithm to draw planar support inductively}
    \label{alg:draw}  
    \begin{algorithmic}[1]


      
    \If{$B_x \subset A_x$}

    \State $E \gets E \cup \{(A,B^{\uparrow})\}$\;

    

    \ForAll{$R_i \in \RR^{(2)}_A$}
    
    \State \textbf{rename} $f(R_i)$ to $R_i$\;
    
    \EndFor

    
    \State \Return $G_P = (\RR^K(P), E)$\;
     
    \EndIf
    
\If{$\RR^{(1)}_A \neq \emptyset$}

    \State $i \gets 1$\;

    \While{$i < |\RR^{(1)}_A|$}
    
    \State $E \gets E \cup \{(R_i, R_{i+1})\}$\; 
    \label{alg:type1-path}

    \State $E \gets E \cup \{(R_i, f(R_i))\}$\;
    \label{alg:type1-image}
    
        \State $i \gets i + 1$\;
    
    \EndWhile

    \State $E \gets E \cup \{(R_i, f(R_i))\}$\;
      


    \ForAll{$R_i \in \RR^{(2)}_A$}
    
    \State \textbf{rename} $f(R_i)$ to $R_i$\;
    
    \EndFor

\State \Return $G_P = (\RR^K(P), E)$\;
     
\EndIf

\If{$\RR^{(2)}_A \neq \emptyset$}

    \ForAll{$(R_i,R_j) \in \TT(\RR_A)$  where $R_i$ is the parent of $R_j$ in $\TT(\RR_A)$, $R_j \in \RR^{(2)}_A$ and $R_i \in \RR^{(1)}_A$}

    \State $E \gets E \setminus \{(f(R_i), f(R_j))\}$\;
    \label{alg:del}
    \State $E \gets E \cup \{(R_i,f(R_j))\}$\;
    \label{alg:add}

    \EndFor

    \ForAll{$R_i \in \RR^{(2)}_A$}
      
        \State \textbf{rename} $f(R_i)$ to $R_i$\;
    \EndFor
      \EndIf
      
      
\State \Return $G_P = (\RR^K(P), E)$\;

\end{algorithmic}
\end{algorithm}

\begin{figure}[h]
  \centering
  \includegraphics[width=0.75\textwidth]{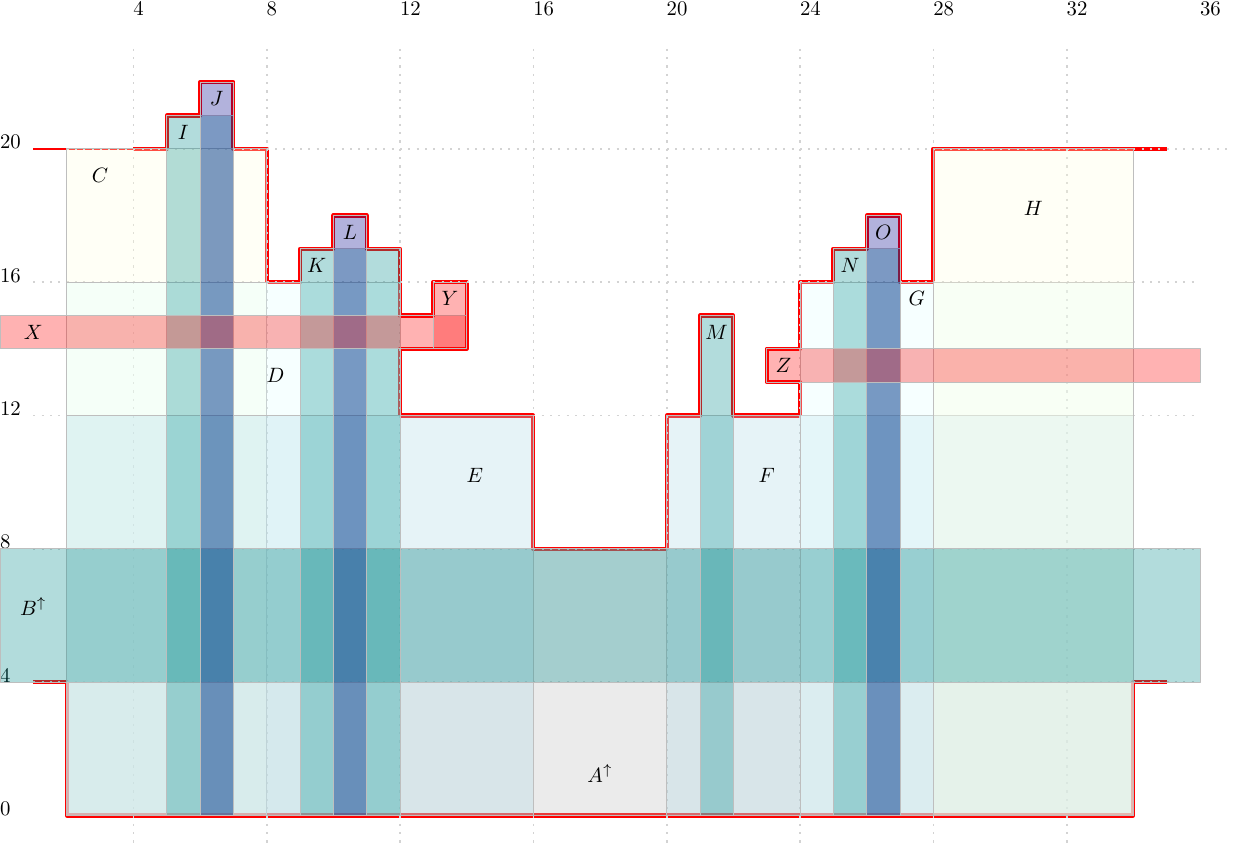}
  \caption{An example polygon illustrating the Type-1 and Type-2 rectangles. 
  }
  \label{fig:type-12}
\end{figure}

\begin{figure}[h]
  \centering
  \includegraphics[width=0.75\textwidth]{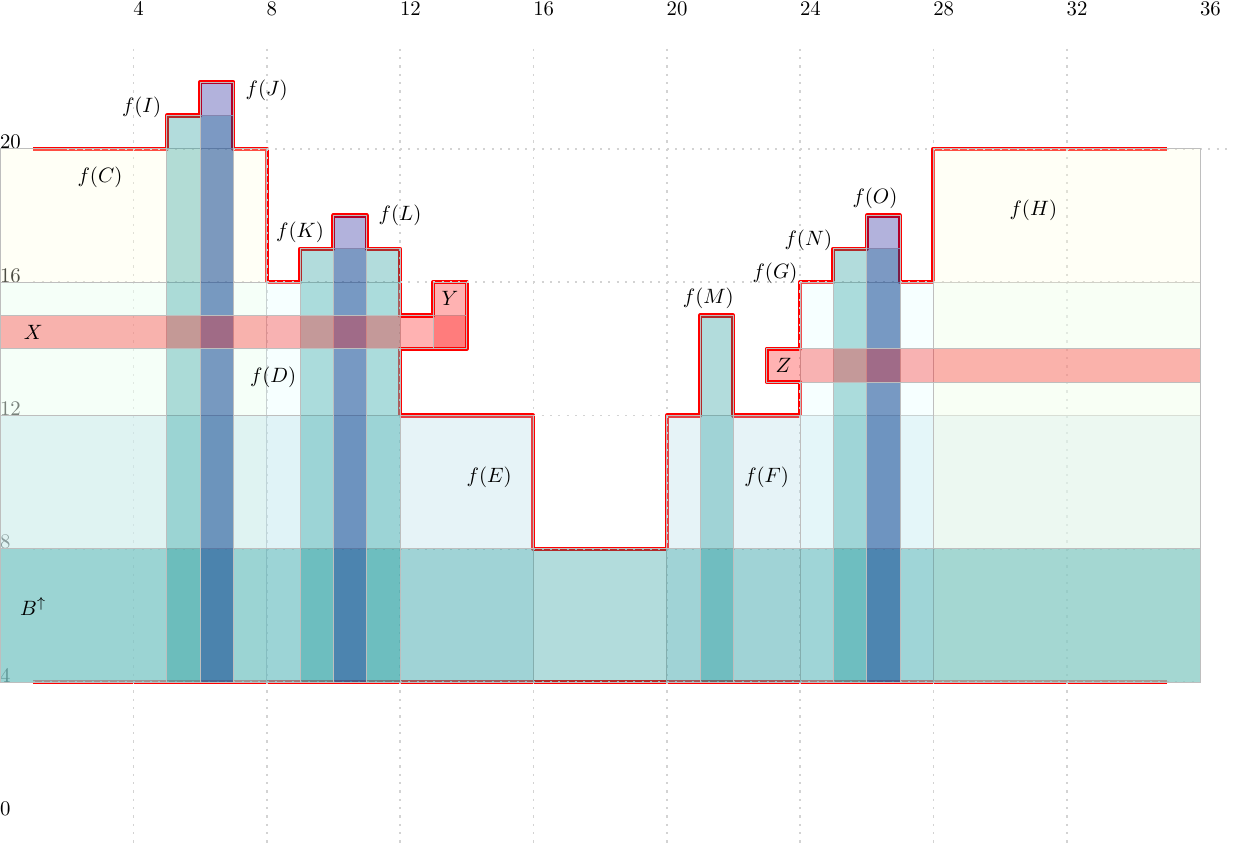}
  \caption{For every Type-1 or Type-2 rectangle $R$ in Figure~\ref{fig:type-12} there is a $f(R)$ in $P'$. 
  }
  \label{fig:type-12-fP}
\end{figure}

\subsection{
An example Illustrating the Algorithm~\ref{alg:draw}.}

   The Type-1 rectangles are:
    $A^\uparrow = [2,34]\times[0,8]$,
    $C = [2,8]\times[0,20]$,
    $D = [2,12]\times[0,16]$,
    $E = [2,16]\times[0,12]$,
    $F = [20,34]\times[0,12]$,
    $G = [24,34]\times[0,16]$,
    $H = [28,34]\times[0,20]$.

    The Type-2 rectangles are:
    $B^\uparrow = [0,36]\times[4,8]$,
    $I = [5,7]\times[0,21]$,
    $J = [6,7]\times[0,22]$,
    $K = [9,12]\times[0,17]$,
    $L = [10,11]\times[0,18]$,
    $M = [21,22]\times[0,15]$,
    $N = [25,27]\times[0,17]$,
    $O = [26,27]\times[0,18]$.

    The images of the Type-1 and Type-2 rectangles are:
   
    $f(C) =  [0,8]\times[4,20]$,
    $f(D) =  [0,12]\times[4,16]$,
    $f(E) =  [0,16]\times[4,12]$,
    $f(F) =  [20,26]\times[4,12]$,
    $f(G) =  [24,36]\times[4,16]$,
    $f(H) =  [28,26]\times[4,20]$,

    $f(I) = [5,7]\times[4,21]$,
    $f(J) = [6,7]\times[4,22]$,
    $f(K) = [9,12]\times[4,17]$,
    $f(L) = [10,11]\times[4,18]$,
    $f(M) = [21,22]\times[4,15]$,
    $f(N) = [25,27]\times[4,17]$,
    $f(O) = [26,27]\times[4,18]$.
    
    The so-called red rectangles are:
    $X =  [0,14]\times[14,15]$,
    $Y =  [13,14]\times[14,16]$,
    $Z =  [23,36]\times[13,14]$.
    
\begin{figure}[h]
  \centering
  \includegraphics[width=0.75\textwidth]{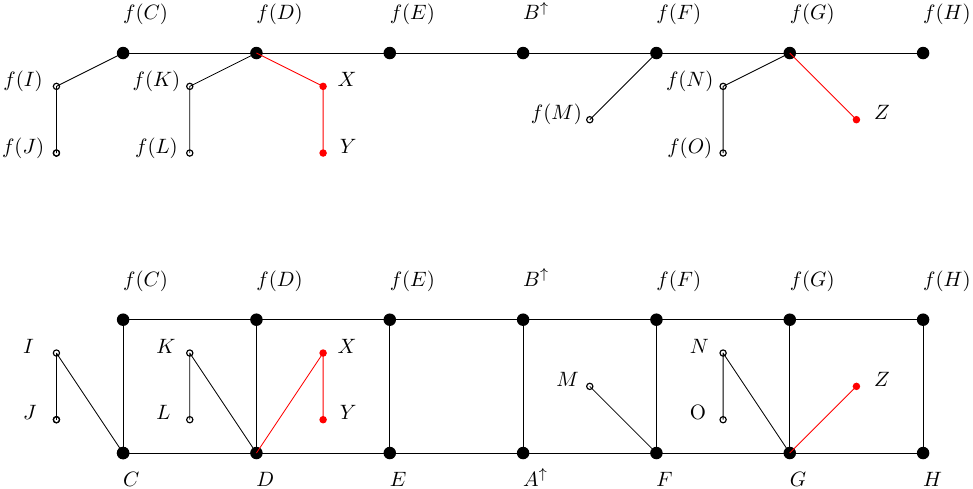}
  \caption{At the top we have the input graph to the algorithm and at the bottom we have the output of the algorithm.
    The vertices corresponding to the Type-1 rectangles and their images are drawn with solid black disks whereas, the vertices corresponding the Type-2 rectangles and their images are drawn with smaller black circles.
    The vertices in red correspond to the red rectangles.}
  \label{fig:graph-type-12}
\end{figure}

\begin{lemma}
  \label{lem:comp-support}
  $G$ is a support graph of $(\RR^K(P), \dP)$. 
\end{lemma}

\begin{proof}
  From the inductive hypothesis, we know that $G_{P'}$ is a support graph of $(\RR^K(P'), \dP')$. There are the following cases:
  \begin{enumerate}
  \item $B_x \subset A_x$.

    $E(G_{P'}) \cup \{(A,B^\uparrow)\}$ ensures that every point in $A\cap \dP$ is supported.
    
  \item $A_x \subset B_x$.

    As $G_{P'}$ is a support graph for $(\RR^K(P'), \dP')$, we only need to ensure that all points in the set
    $\{R \cap \dP \mid R \in \RR_A \}$ are supported by $G_{P}$. Note that $\{R \cap \dP \mid R \in \RR_A \}$ is the set of all blockers of every $R \in \RR_A$. Also, recall that $b_A$ is the set of all bottom blockers of $R \in \RR_A$. Further observe that all the edges in $\TT(\RR_A)$ are included in $G_P$. This is because, edges between two Type-1 vertices are added in Step~\ref{alg:type1-path}, edges between Type-1 and Type-2 vertices are added in Step~\ref{alg:add}, and edges between two Type-2 vertices are present as their images in $\RR_B$ are renamed, in Algorithm~\ref{alg:draw}.
    Hence, $b_A$ is supported by $G_P$.

    Next, we claim that all the top blockers of $R \in \RR_A$ are also supported. If $R$ is a Type-1 vertex then the top blocker is supported because the edge $(R,f(R))$ is added in Step~\ref{alg:type1-image}. If furthermore, $R$ is the top extension of $A$, then one or two reflex vertices of $P$ incident with the corners of $A$. In that case, those reflex vertices are also supported by $(A^{\uparrow}, f(A^{\uparrow}))$.
    If $R$ is a Type-2 vertex then its top blocker is satisfied once $f(R)$ is relabeled as $R$ in $G_{P}$ as the top blocker of $f(R)$ was satisfied in $G_{P'}$.

    Finally, we claim that all the left and right blockers of $R \in \RR_A$ are also supported by $G_P$. If $R$ is a Type-2 vertex then its left and right blockers are satisfied for the same reason as above. If $R$ is a Type-1 vertex and it is also the top extension of $A$, i.e, $R =A^\uparrow$ then its left blocker is $\ell_A$ and its right blocker is $r_A$. Then $\ell_A$ is satisfied by all Type-1 vertices to the left of $A^\uparrow$ and including $A^\uparrow$ in the path $\PP$ of the Type-1 vertices (as described in the proof of Lemma~\ref{lem:path-supp}). Conversely, $r_A$ is satisfied by all the Type-1 vertices to the right of $A^\uparrow$ in $\PP$, including $A^\uparrow$. This path is formed iteratively by Step~\ref{alg:type1-path} of the algorithm. If $R$ is a Type-1 vertex whose left blocker is not $\ell_A$, then such a left blocker is satisfied by the edge $(R,f(R))$, which is added at Step~\ref{alg:type1-image} of the algorithm.

    
  \end{enumerate}
  
\end{proof}



\begin{lemma}
  $G_P$ is planar.
  \label{lem:comp-planar}
\end{lemma}

\begin{proof}


  If $B_x \subset A_x$ then $G_P$ is planar because a single edge $(A, B^{\uparrow})$ is added to $G_{P'}$.

  Let $A_x \subset B_x$.
  If $\RR^{(1)}_A \neq \emptyset$ then edges $(R_i, R_{i+1})$, for $i < |\RR^{(1)}_A|$ can be drawn without crossing as they form a path.
  Also note, the edges $(R_i, f(R_i))$, for $R_i \in \RR^{(1)}_A$ can also be drawn without crossing, because $f(R_i)$ are vertically-blocked blocked rectangles of $P'$ and by inductive hypothesis the vertices corresponding to them are drawn on the outer face of $G_{P'}$, and if $R_i <_1 R_j$ then $f(R_i) <_1 f(R_j)$. See Figures~\ref{fig:type-12} and~\ref{fig:graph-type-12}.

  If $\RR^{(2)}_A \neq \emptyset$ and Lines~\ref{alg:del} and \ref{alg:add} of Algorithm \ref{alg:draw} are executed, then from Lemma~\ref{lem:struct} planarity is ensured.
\end{proof}

\begin{lemma}
  \label{lem:comp-outer}
  Every vertically blocked rectangle in $\RR^K(P)$ lies on the outer face of the plane graph returned by Algorithm~\ref{alg:draw}.
\end{lemma}

\begin{proof}
  The Type-1, Type-2 and the red vertices are vertically blocked and have been (or can be) placed on the outer face of the drawing of $G_P$. See Figure~\ref{fig:graph-type-12}, although the Type-2 and the red vertices are not drawn on the outer face but it can be easily drawn.
  The images of Type-1 vertices which were vertically blocked in $P'$ may cease to be so in $P$. Suppose there exist a rectangle $R$, an image of a Type-1 vertex, that is vertically blocked in $P$ then we claim that $R$ is on the outer face of $G_P$. If $R$ is vertically blocked in $P$ then it has two distinct vertical blockers in $P'$. From Lemma~\ref{lem:vert-blocker}, it is a cut vertex. Hence, even after adding a layer of Type-1 vertices to $G_{P'}$, $R$ remains on the outer face of $G_P$.
\end{proof}

\end{document}